\documentclass[acmsmall]{acmart}

\setcopyright{cc}
\setcctype{by}
\acmDOI{10.1145/3808315}
\acmYear{2026}
\acmJournal{PACMPL}
\acmVolume{10}
\acmNumber{PLDI}
\acmArticle{237}
\acmMonth{6}
\acmSubmissionID{pldi26main-p482-p}
\received{2025-11-14}
\received[accepted]{2026-04-03}

\usepackage{graphicx}
\usepackage{amsmath}
\usepackage{amsthm}
\usepackage{xcolor}
\usepackage{hyperref}
\usepackage{booktabs}
\usepackage{caption}
\usepackage{subcaption}
\usepackage{tabularx}
\usepackage[mathscr]{euscript}
\usepackage{xspace}
\usepackage[ruled,commentsnumbered,linesnumbered,noend]{algorithm2e}
\usepackage{tikz}
\usepackage[inline]{enumitem}
\usepackage[super]{nth}
\usepackage{pgfplots}
\usepackage{tcolorbox}
\usepackage{thm-restate}
\usepackage{listings}
\usepackage{wrapfig}
\usepackage{dashrule}
\usepackage{comment}

\usepackage{lipsum}
\usepackage{dutchcal}
\usepackage{multirow}

\let\euscr\mathscr
\graphicspath{ {images/} }

\pgfplotsset{width=10cm,compat=1.9}
\usepgfplotslibrary{statistics}

\theoremstyle{definition}
\newtheorem{theorem}{Theorem}[section]
\newtheorem{definition}[theorem]{Definition}
\newtheorem{corollary}[theorem]{Corollary}
\newtheorem{lemma}[theorem]{Lemma}
\newtheorem{proposition}[theorem]{Proposition}

\newtheorem{problem}{Problem}
\newtheorem*{problem*}{Problem}
\newtheorem*{remark*}{Remark}

\newenvironment{myparagraph}[1]{\vspace{0.1in}\noindent{\bf #1.}}{}

\definecolor{specBoxOutlineColor}{rgb}{0.122, 0.435, 0.698}
\definecolor{specBoxBgColor}{rgb}{1.000, 0.997, 0.990} 
\newcommand{\specBox}[1]{%
\begin{tcolorbox}[colframe=specBoxOutlineColor,colback=specBoxBgColor,boxrule=0.5pt,arc=4pt,
      left=6pt,right=6pt,top=2pt,bottom=2pt,boxsep=0pt,width=\columnwidth,fontupper=\footnotesize]%
      {\emph{#1}}
\end{tcolorbox}%
}

\newcommand{\figlabel}[1]{\label{fig:#1}}
\newcommand{\figref}[1]{Fig.~\ref{fig:#1}}
\newcommand{\seclabel}[1]{\label{sec:#1}}
\newcommand{\secref}[1]{Section~\ref{sec:#1}}

\newcommand{\deflabel}[1]{\label{def:#1}}
\newcommand{\defref}[1]{Definition~\ref{def:#1}}

\newcommand{\thmlabel}[1]{\label{thm:#1}}
\newcommand{\thmref}[1]{Theorem~\ref{thm:#1}}
\newcommand{\proplabel}[1]{\label{prop:#1}}

\newcommand{\lemlabel}[1]{\label{lem:#1}}
\newcommand{\lemref}[1]{Lemma~\ref{lem:#1}}
\newcommand{\corlabel}[1]{\label{cor:#1}}
\newcommand{\corref}[1]{Corollary~\ref{cor:#1}}

\newcommand{\applabel}[1]{\label{app:#1}}
\newcommand{\appref}[1]{Appendix~\ref{app:#1}}
\newcommand{\algolabel}[1]{\label{app:#1}}
\newcommand{\algoref}[1]{Algorithm~\ref{app:#1}}
\newcommand{\linelabel}[1]{\label{line:#1}}
\newcommand{\lineref}[1]{Line~\ref{line:#1}}

\newcommand{\set}[1]{\{#1\}}
\newcommand{\setpred}[2]{\set{#1 \,|\, #2}}
\newcommand{\tuple}[1]{\langle#1\rangle}
\newcommand{\proj}[2]{#1|_{#2}}
\renewcommand{\emptyset}{\varnothing}

\newcommand{\nats}{\mathbb{N}}
\newcommand{\rats}{\mathbb{Q}}

\newcommand{\true}{{\sf true}}
\newcommand{\false}{{\sf false}}
\newcommand{\pass}{\sf ok}
\newcommand{\fail}{\sf fail}

\newcommand{\procs}{\euscr{P}}
\newcommand{\vals}{\euscr{V}}
\newcommand{\methods}{\euscr{M}}
\newcommand{\Times}{\euscr{T}}
\newcommand{\alphabet}{\Sigma}

\newcommand{\hist}{H}
\newcommand{\lin}{\ell}

\newcommand{\methodAttr}{{\sf m}\xspace}
\newcommand{\procAttr}{{\sf p}\xspace}
\newcommand{\valAttr}{{\sf v}\xspace}
\newcommand{\argAttr}{{\sf arg}\xspace}
\newcommand{\retAttr}{{\sf ret}\xspace}
\newcommand{\invTimeAttr}{{\sf inv}\xspace}
\newcommand{\resTimeAttr}{{\sf res}\xspace}

\newcommand{\methodOf}[1]{\methodAttr(#1)}
\newcommand{\procOf}[1]{\procAttr(#1)}
\newcommand{\valOf}[1]{\valAttr(#1)}
\newcommand{\argOf}[1]{\argAttr(#1)}
\newcommand{\returnOf}[1]{\retAttr(#1)}
\newcommand{\invTimeOf}[1]{\invTimeAttr(#1)}
\newcommand{\resTimeOf}[1]{\resTimeAttr(#1)}

\newcommand{\opr}[3]{{#1}_{#3}\texttt{(}#2\texttt{)}}
\newcommand{\abs}[1]{{\sf abs}(#1)}

\definecolor{ADTSetColor}{rgb}{0.0, 0.525, 0.043}
\newcommand{\ADT}{\mathcal{D}}
\newcommand{\registerDS}{\texttt{\textcolor{ADTSetColor}{register}}\xspace}
\newcommand{\counterDS}{\texttt{\textcolor{ADTSetColor}{counter}}\xspace}
\newcommand{\mutexDS}{\texttt{\textcolor{ADTSetColor}{mutex}}\xspace}
\newcommand{\setDS}{\texttt{\textcolor{ADTSetColor}{set}}\xspace}
\newcommand{\stackDS}{\texttt{\textcolor{ADTSetColor}{stack}}\xspace}
\newcommand{\sizeStackDS}{\texttt{\textcolor{ADTSetColor}{size-visible-stack}}\xspace}
\newcommand{\queueDS}{\texttt{\textcolor{ADTSetColor}{queue}}\xspace}
\newcommand{\pqueueDS}{\texttt{\textcolor{ADTSetColor}{priority-queue}}\xspace}

\newcommand{\multisetDS}{\texttt{\textcolor{ADTSetColor}{multiset}}\xspace}
\newcommand{\mapDS}{\texttt{\textcolor{ADTSetColor}{map}}\xspace}

\newcommand{\faaDS}{\texttt{\textcolor{ADTSetColor}{faa-register}}\xspace}
\newcommand{\casDS}{\texttt{\textcolor{ADTSetColor}{cas-register}}\xspace}
\newcommand{\semDS}{\texttt{\textcolor{ADTSetColor}{semaphore}}\xspace}

\newcommand{\stackDSUpper}{\texttt{\textcolor{ADTSetColor}{Stack}}\xspace}
\newcommand{\queueDSUpper}{\texttt{\textcolor{ADTSetColor}{Queue}}\xspace}

\newcommand{\spec}{\mathbb{T}}
\newcommand{\RegisterSpec}{{\spec}_{\registerDS}}

\newcommand{\StackSpec}{{\spec}_{\stackDS}}
\newcommand{\QueueSpec}{{\spec}_{\queueDS}}

\definecolor{methodNameColor}{rgb}{0.122, 0.435, 0.698}
\newcommand{\methodName}[1]{{\tt \textcolor{methodNameColor}{#1}}\xspace}
\definecolor{valueColor}{rgb}{0.698, 0.111, 0.111}
\newcommand{\val}[1]{{\tt\small \textcolor{valueColor}{#1}}\xspace}

\newcommand{\vtrue}{\val{\top\!\!\!\top}}
\newcommand{\vfalse}{\val{\bot\!\!\!\bot}}
\newcommand{\vok}{\val{\pass}}
\newcommand{\vfail}{\val{\fail}}
\newcommand{\methodsRegisters}{\methods_{\registerDS}}
\newcommand{\mread}{\methodName{read}}
\newcommand{\mwrite}{\methodName{write}}

\newcommand{\methodsCounter}{\methods_{\counterDS}}
\newcommand{\inc}{\methodName{inc}}
\newcommand{\dec}{\methodName{dec}}
\newcommand{\get}{\methodName{get}}
\newcommand{\methodsQueue}{\methods_{\queueDS}}
\newcommand{\methodsPQueue}{\methods_{\pqueueDS}}
\newcommand{\enq}{\methodName{enq}}
\newcommand{\deq}{\methodName{deq}}
\newcommand{\peek}{\methodName{peek}}

\newcommand{\methodsStack}{\methods_{\stackDS}}
\newcommand{\methodsSizeStack}{\methods_{\sizeStackDS}}
\newcommand{\push}{\methodName{push}}
\newcommand{\pop}{\methodName{pop}}
\newcommand{\methodsSet}{\methods_{\setDS}}

\newcommand{\contains}{\methodName{has}}
\newcommand{\remove}{\methodName{del}}

\newcommand{\methodsMutex}{\methods_{\mutexDS}}

\newcommand{\methodsMultiset}{\methods_{\multisetDS}}
\newcommand{\methodsMap}{\methods_{\mapDS}}

\newcommand{\methodsCAS}{\methods_{\casDS}}
\newcommand{\methodsFAA}{\methods_{\faaDS}}
\newcommand{\methodsSem}{\methods_{\semDS}}

\newcommand{\add}{\methodName{add}}

\newcommand{\mapPut}{\methodName{put}}

\newcommand{\cas}{\methodName{cas}}

\newcommand{\faa}{\methodName{faa}}

\newcommand{\acq}{\methodName{acq}}
\newcommand{\rel}{\methodName{rel}}

\newcommand{\keys}{\euscr{K}}
\newcommand{\key}{\mathcal{k}}
\newcommand{\dom}{\mathop{\mathrm{dom}}}
\newcommand{\visited}{{\tt visited}}

\newsavebox{\figholder}

\newenvironment{scatterplot}[2]{%
  \begin{tikzpicture}
  \begin{axis}[
    width=\linewidth,
    xlabel={#1},
    ylabel={#2},
    grid=both,
    grid style={dashed,gray!30},
    legend style={at={(0.02,0.98)}, anchor=north west, font=\small},
    mark options={scale=1.2}
  ]
}{%
  \end{axis}
  \end{tikzpicture}
}

\newcommand{\addscatter}[3]{%
  \addplot+[
    x filter/.code=\pgfmathparse{\pgfmathresult + 10}, 
    only marks,
    mark size=1.2pt,
    #3
  ] table[
    x index=0,
    y index=1,
    col sep=space
  ] {#1};
  \addlegendentry{#2};
}

\newcommand{\np}{\textsf{NP}}
\newcommand{\poly}{\textsf{poly}}
\newcommand{\fpt}{\textsf{FPT}}

\newcommand{\fptlin}{\texttt{FPTLin}\xspace}

\newcommand{\lbl}[1]{\textsf{lbl}(#1)}
\newcommand{\histpo}[1]{<_{#1}}

\newcommand{\idealsOf}[1]{\textsf{Ideals}_{#1}}
\newcommand{\frontiersOf}[1]{\mathcal{F}_{#1}}
\newcommand{\fg}[1]{G_{#1}}     
\newcommand{\fgnodes}[1]{N_{#1}}     
\newcommand{\fgedges}[1]{E_{#1}}     
\newcommand{\init}{\textsf{init}}
\newcommand{\final}{\textsf{final}}

\newcommand{\match}{{\tt match}}

\newcommand{\partstate}{\mathcal{S}}
\newcommand{\ifunc}{\mathcal{I}}

\newcommand{\inferrule}[2]{\frac{\begin{array}{c}#1\end{array}}{\begin{array}{c}#2\end{array}}}

\newcommand{\ltstrans}[1]{\xrightarrow{#1}}

\newcommand{\lang}[1]{\mathcal{L}(#1)}
\newcommand{\streach}[1]{{\tt Reach}_\stackDS(#1)}
\newcommand{\transclo}[1]{{\tt TransClo}_\stackDS(#1)}
\newcommand{\matmul}[1]{{\tt MatMul}_\stackDS(#1)}
\newcommand{\boolmatmul}[1]{{\tt BMM}(#1)}

\newcommand{\qbot}{\bot}
\newcommand{\IdlRel}[1]{Q^{\sf ideals}_{#1}}
\newcommand{\queueStep}{\mathrel{\leadsto_{\queueDS}}}

\newcommand*\circledsmall[1]{\tikz[baseline=(char.base)]{
  \node[shape=circle,draw=black,fill=orange!20!white,inner sep=0.5pt, solid] (char) {\textcolor{black}{\tt\footnotesize{#1}}};}}

\SetDataSty{textsc}

\SetKwProg{myfun}{function}{}{}
\SetKwProg{myproc}{procedure}{}{}
\SetKwProg{myhandler}{handler}{}{}
\SetKwFunction{OODTLin}{OODTLin}
\SetKwFunction{StackLin}{CFLLin}
\SetKwFunction{QueueLin}{QueueLin}

\SetKwFunction{DFS}{OODTLinDFS}

\newcommand{\repr}{\textsf{Rep}}
\newcommand{\nonTerminals}[1]{\mathsf{NT}_{#1}}
\newcommand{\nullableNonTerminals}[1]{\mathsf{NT}^{\emptyset}_{#1}}
\newcommand{\prodRules}[1]{\rightarrow_{#1}}

\SetKwInput{Input}{Input}
\SetKwInOut{Output}{Output}
\SetKw{Let}{let}
\SetKw{Break}{break}
\SetKw{Continue}{continue}
\SetKw{AND}{and}
\SetKw{WHERE}{where}
\SetKw{OR}{or}
\SetKw{NOT}{not}
\SetKw{declare}{declare}
\SetKw{report}{report}
\SetKw{exit}{exit}
\SetKwFor{Foreach}{for each}{}{}%
\SetKwFor{RepTimes}{repeat}{times}{end}
\SetKwData{NULL}{null}
\DontPrintSemicolon

\begin{document}

\title{Fixed Parameter Tractable Linearizability Monitoring}

\author{Zheng Han Lee}
\orcid{0009-0000-7130-2493}
\affiliation{%
  \institution{National University of Singapore}
  \city{Singapore}
  \country{Singapore}
}
\email{zhlee@u.nus.edu}

\author{Umang Mathur}
\orcid{0000-0002-7610-0660}
\affiliation{%
  \institution{National University of Singapore}
  \city{Singapore}
  \country{Singapore}
}
\email{umathur@nus.edu.sg}


\begin{abstract}
We study the linearizability monitoring problem, which asks whether a given concurrent history of a data structure is equivalent to some sequential execution of the same data structure. In general, this problem is $\np$-hard, even for simple objects such as registers. Recent work has identified tractable cases for restricted classes of histories, notably unambiguous and differentiated histories.

We revisit the tractability boundary from a fine-grained, parameterized perspective. We show that for a broad class of data structures --- including stacks, queues, priority queues, and maps—linearizability monitoring is fixed-parameter tractable when parameterized by the number of processes. Concretely, we give an algorithm running in time $O(c^{k} \cdot \poly(n))$, where $n$ is the history size, $k$ is the number of processes, and $c$ is a constant, yielding efficient performance when $k$ is small.
Our approach reduces linearizability monitoring to a language reachability problem on graphs, which asks whether a labeled graph admits a path whose label sequence belongs to a fixed language $L$. We identify classes of languages that capture the sequential specifications of the above data structures and show that language reachability is efficiently solvable on the graph structures induced by concurrent histories.

Our results complement prior hardness results and existing tractable subclasses, and provide a unified algorithmic framework. We implement our approach and demonstrate significant runtime improvements over existing algorithms, which exhibit exponential worst-case behavior.
\end{abstract}


\maketitle


\section{Introduction}
\seclabel{intro}

Linearizability, originally proposed by Herlihy and Wing~\cite{Wing1990},
serves as the standard correctness criterion for implementations of concurrent
data structures. Conceptually, it asks: given a concurrent implementation of an
abstract data type (ADT), are all of its behaviors equivalent to those of an
ideal sequential implementation? Full-fledged formal verification of
linearizability is undecidable in general~\cite{Bouajjani2013}, and the known
decidable classes of programs and specifications are rather restricted.

The focus of this paper is the more pragmatic \emph{linearizability monitoring} problem. 
Instead of asking that all behaviors of an implementation
be correct, linearizability monitoring asks the more modest question: given a
single execution history $\hist$ produced by running a concurrent
implementation of an ADT, is $\hist$ equivalent to some execution of a
sequential implementation? 
Linearizability monitors form a core component of stress
testing and stateless model checking, where implementations are exercised
repeatedly and the resulting histories are checked offline for violations.
Efficient algorithms for the monitoring problem translate directly into more
effective testing and exploration, and are also the subject of study in this work.

The complexity of monitoring depends on both the ADT and the shape of the histories. 
In~\cite{Gibbons1997} Gibbons and Korach studied the complexity of the monitoring problem for one
of the simplest \registerDS ADTs where objects expose $\mwrite$ and $\mread$ operations.
They showed that, in general, the monitoring
problem for \registerDS is $\np$-complete, and further showed that the problem
is polynomial-time solvable when the input history additionally associates
each $\mread$ operation to a unique $\mwrite$ operation, or equivalently,
under the restriction that each value is written to at most once in the history.
This \emph{unambiguity} restriction,
was recently shown to also yield polynomial time monitoring algorithms for other data structures
including \stackDS{}s, \queueDS{}s and \pqueueDS{}s~\cite{lee2025,Abdulla2025}.
Intuitively, the unambiguity restriction (or \emph{data differentiatedness})
applies to a history of the above ADTs if each value is added (and removed) exactly once.

The motivation behind both the above restrictions stem from the
idea that many ADTs enjoy \emph{data independence}, i.e., if a history is admitted,
then it continues to be admitted even after substituting values in it with other values.
While ADTs such as \stackDS{}s, \queueDS{}s and \pqueueDS{}s satisfy data independence,
it is unclear if this property holds beyond a small class of ADTs.
More generally, the unambiguity restriction presupposes a clear demarcation between `add'
and `remove' operations~\cite{lee2025}, which are difficult to demarcate for many
of ADTs, including \counterDS{}s, \faaDS{}s or \sizeStackDS{}s (see \figref{oodt-specs} in \secref{oodt-lin}),
where values stored in the abstract state need not correspond directly to values appearing in operations.
Finally, even when the sequential specification of the ADT observes data independence,
it may not always be easy to modify the actual source code of the concurrent implementation
to make it generate unambiguous histories, for practical reasons such as
to avoid degradation of performance, or simply because the source code may not be available at the time of testing.

The theme of this work is to design tractable algorithms for linearizability monitoring
of a large class of ADTs that work for arbitrary histories.
Some observations towards this are in order.
First, for the \setDS{} and \multisetDS{} ADTs, the monitoring algorithms
proposed in~\cite{Abdulla2025} do not assume unambiguity restriction and work in polynomial time.
Second, $\np$-hardness results for the case of other simple ADTs such as \registerDS{}s~\cite{Gibbons1997} 
and subsequently for \stackDS{}, \queueDS{} and \pqueueDS{}~\cite{Gibbons2002LinJournal} 
forbid the existence of fully polynomial time algorithms (unless $\sf{P} = \sf{NP}$).
Nonetheless, Gibbons and Korach~\cite{Gibbons1997} showed that,
monitoring for \registerDS{} can be performed in time $O(2^k \cdot n)$,
where $n$ is the total number of operations in the history and $k$
is the number of processes involved.
That is, intractability of monitoring for \registerDS{} arises
primarily when the number of processes in the history scales arbitrarily.
In other words, linearizability monitoring for \registerDS{} becomes 
tractable when the number of processes is fixed, i.e., 
it is FPT in the parameter $k$ (number of processes).
In this work, we further the parameterized complexity landscape, 
initiated by Gibbons and Korach for \registerDS,
and ask \emph{for what class of ADTs is linearizability monitoring
FPT in the number of processes?}
In the following, we discuss our contributions 
\circledsmall{C1} --- \circledsmall{C5} towards answering this question.

\myparagraph{ADT classes for FPT linearizability monitoring}
Towards the above goal, we identify three broad families of ADTs for which 
linearizability monitoring is FPT.
The first class of ADTs we identify is that of \emph{order-oblivious} data types (OODTs),
and their parametric extension \emph{$\alpha$-order-oblivious} data types ($\alpha$-OODTs, $\alpha \in \nats$) (\circledsmall{C2}).
Intuitively, these are ADTs for which the abstract state that a given sequence
of operation results in depends primarily upon the multiset of these operations,
and does not depend on the precise order except for
a small finite amount of order-sensitive control; the exact amount is determined by
parameter $\alpha$.
We show that for an $\alpha$-OODT $\ADT$, linearizability monitoring
can be performed in time proportional to $O(\alpha{}nk2^k)$ for histories with $n$
operations over $k$ processes (\thmref{alpha-OODT-FPT}).
The second class is that of context-free ADTs (\circledsmall{C3}),
whose sequential specifications are given using a context free grammar, and includes the $\stackDS{}$ ADT.
We show that for such ADTs, linearizability can be monitored in time
proportional to $O(g_\hist \cdot n^3 \cdot 2^{3k})$ (\thmref{CFL-ADT-FPT}).
Here, $g_\hist$ denotes the size of the grammar
projected to the operations of the given history $\hist$.
Furthermore, we show that this running time can be optimized
using fast matrix multiplication for $\stackDS$ (\thmref{stack-FPT}).
Third, we study \queueDS (\circledsmall{C4}), whose specification is
neither order-oblivious nor context-free and show that linearizability can be monitored
for \queueDS in $O(k2^{2k}\cdot n^2)$ time (\thmref{queue-FPT}).

\myparagraph{Linearizability as language reachability}
The inspiration behind the above ADT classes and their linearizability monitoring,
in fact, comes from the FPT algorithm
for that of \registerDS,
initially proposed by Gibbons and Korach~\cite{Gibbons1997}.
At a high level, this algorithm constructs a 
directed acyclic graph,
whose nodes correspond to subsets of operations 
and whose edges correspond to addition of a single new operation that 
(a) would not violate the semantics of $\registerDS{}$ and 
(b) be consistent with the happens-before order
induced by the history.
Then, the linearizability monitoring problem can be viewed as the problem of reaching 
a designated sink node
from a designated source node in this graph.
We insist that this algorithm can be lightly reformulated 
to decouple the structure induced by the history from that induced by the ADT as follows.
First, the graph edges only need to account for consistency with happens-before of this history;
this is often referred to as the \emph{frontier graph} of this history, since nodes
represent the different \emph{frontiers} that correspond to different allowable
prefixes of the history.
Second, instead of asking if there is any path from the source to the sink, we
ask if there is a path in the frontier graph 
that is also labeled with a sequence of operations that is legal as per the
sequential specification of the \registerDS{} ADT.
This alternative formalism is an instance of the \emph{language reachability} problem,
which in general, is parametrized by a formal language $L$
and asks, given an input directed labeled graph $G$, is there a walk
from the source vertex $s$ to the target vertex $t$ that is labeled with a word
that belongs to $L$.
In this work, we further this reformulation and approach 
linearizability monitoring for arbitrary ADTs from the unifying
standpoint of language reachability:
the graph in question is the frontier graph of the history
and the language in question is the sequential specification $\spec{}_{\ADT}$
of the ADT $\ADT{}$ (\circledsmall{C1}).
This reduction factors the monitoring problem
into a graph component, which is uniform and well understood, and a language
component, which depends only on the sequential
specification~\cite{Melski1997,koutris2023fine,Conrado2025}.

The frontier graph of a history with $n$ operations on $k$ processes has size $O(nk2^k)$.
FPT algorithms for linearizability monitoring can thus be obtained
for ADTs $\ADT$ for which the language reachability problem against their sequential specification
$\spec_{\ADT}$ can be solved in polynomial time in the size of the graph.
For the class of $\alpha$-OODTs (and OODTs), we show that language reachability can be solved
in polynomial time (assuming membership checking can be solved in polynomial time).
The class of CFL ADTs are those whose specifications are context free languages, and the language reachability problem
for them is already known to be solvable in polynomial time.
The most interesting case is that of the $\queueDS{}$ ADT.
As such, the generic language reachability problem for them is known to be undecidable for arbitrary graphs.
Even for the smaller class of directed acyclic graphs (DAGs), $\QueueSpec{}$-language
reachability is known to be
$\np$-hard, and we show that this holds even for bounded pathwidth DAGs.
Nonetheless, we prove that, for the class of frontier graphs arising in the context
of linearizability monitoring, a
specialized rewriting system over triples of frontiers yields a polynomial-time
algorithm for language reachability, and thus an \fpt{} monitor.

\myparagraph{Implementation and evaluation}
We implement these algorithms in our tool, \fptlin, and evaluate it on histories
generated from realistic concurrent data-structure implementations. Our
implementation supports \stackDS, \queueDS, \pqueueDS, and a variety of OODTs
and $\alpha$-OODTs \circledsmall{C5}. We compare \fptlin against the state-of-the-art
linearizability monitors LinP~\cite{lee2025}, LiMo~\cite{Abdulla2025}, and VeriLin~\cite{Jia2023} on histories
produced using the Scal benchmarking suite~\cite{Scal2015}. Across a range of
workloads, \fptlin matches or outperforms VeriLin for moderate numbers of
threads and remains robust as concurrency increases, where VeriLin often times
out or exhausts memory. Moreover, \fptlin scales to histories with around one
million operations for \queueDS and \pqueueDS, staying competitive with LinP and LiMo,
while not being restricted to ambiguous histories.

\myparagraph{Organization}
The rest of the article is organized as follows.
After presenting our preliminaries in \secref{prelim},
we discuss, in \secref{lang-reachability}, frontier graphs and show 
the reduction of linearizability monitoring
to language reachability over such graphs.
In \secref{oodt-lin}, \secref{cfl-adt} and \secref{queue-adt}, we discuss FPT monitoring
for respectively $\alpha$-OODTs, CFL ADTs (and \stackDS{})
and the \queueDS{} ADT.
We discuss  the details of the implementation and evaluation of our algorithms in~\secref{eval},
discuss related work in \secref{related} and concluding remarks in \secref{conclusions}.
Proofs of many of our results have been delegated to the appendix.

\section{Preliminaries}
\seclabel{prelim}

Here, we recall the background relevant for describing the linearizability problem.
Expert readers may skip to later sections.

\subsection{Histories and Operations}

The focus of this work is to design algorithms for
linearizability monitoring.
The input to this problem is a \emph{history} that tracks 
execution of a concurrent data structure, consisting of operations 
performed by multiple client processes
invoking the data structure concurrently.

\myparagraph{Operations}
A history consists of operations accessing a concurrent data structure.
Each operation is a tuple of the form 
$o = \tuple{id, p, x, m, v_{\argAttr}, v_{\retAttr}, t_{\invTimeAttr}, t_{\resTimeAttr}}$.
Here, $id$ is a unique identifier for $o$,
$p \in \procs$ denotes the identifier of the process that performs $o$,
$x$ denotes the concurrent object (instance of some ADT $\ADT$) that $o$ accesses,
$m \in \methods_{\ADT}$ denotes the method of the ADT $\ADT$ (we assume $\methods_\ADT$ is finite),
$v_{\argAttr}, v_{\retAttr} \in \vals_{\ADT}$ 
denote the argument and return values of the operation $o$, and
$t_{\invTimeAttr}, t_{\resTimeAttr} \in \rats_{\geq 0}$ 
denote the (rational) time corresponding to the invocation and 
response of the operation $o$;
we require that $t_{\invTimeAttr} < t_{\resTimeAttr}$. 
We will use $\procOf{o}$, $\methodOf{o}$, $\argOf{o}$, 
$\returnOf{o}$, 
$\invTimeOf{o}$ and $\resTimeOf{o}$ 
to denote respectively the process $p$, method $m$, value $v_{\argAttr}$,
value $v_{\retAttr}$, 
invocation time $t_{\invTimeAttr}$ and response time $t_{\resTimeAttr}$ of operation $o$. 
When clear from context, we will drop the identifier $id$.
Since linearizability is a local property~\cite{Wing1990}, it suffices
to only talk about histories and operations of a single object, in which
case we will also drop the object identifier $x$.

\myparagraph{Concurrent and Sequential Histories}
A \emph{concurrent history} $\hist$ is a finite collection of operations. 
For instance, the history $\hist_{\sf stack} = \set{o_1 = \tuple{p_1, \push, \val{3}, \vok, 1, 4}, o_2 = \tuple{p_2, \pop, \_ , \val{3}, 3, 6}, o_3=\tuple{p_1, \pop, \_, \vfail, 5, 7}}$ of a stack object
comprises of three operations $o_1$, $o_2$ and $o_3$.
$o_1$ pushes the value $\val{3}$ in the invocation/response interval $[1,4]$.
$o_2$ is a pop operation  by process $p_2$ that returns value $\val{3}$ in the interval $[3,6]$. 
Finally, $o_3$ is a pop operation
by process $p_1$ in the interval $[5, 7]$ that failed to execute.
We assume histories are well-formed in that, for each process $p \in \procs_\hist$,
the invocation/response intervals of any two operations are disjoint,
where $\procs_\hist$ is the set of processes of operations in $\hist$.
The size $|\hist|$ of a history $\hist$ denotes the number of operations in $\hist$.
We use $\Times_\hist = \bigcup_{o\in \hist} \set{\invTimeOf{o}, \resTimeOf{o}}$ 
to denote the set of invocation and response times in $\hist$.
A history $\hist$ is said to be \emph{sequential} if all time intervals in it are non-overlapping,
i.e., for every pair $o_1 \neq o_2 \in \hist$, we have
either $\resTimeOf{o_1} < \invTimeOf{o_2}$ or
$\resTimeOf{o_2} < \invTimeOf{o_1}$.
A history $\hist$ naturally induces a strict partial order $\histpo{\hist}$
on its operations such that $o_1 \histpo{\hist} o_2$ iff $\resTimeOf{o_1} < \invTimeOf{o_2}$;
$\histpo{\hist}$ is also popularly referred to as the \emph{happens-before}
order of $\hist$.
Observe that if $\hist$ is sequential, then $\histpo{\hist}$ is a total order.


\myparagraph{Linearizations}
A linearization of a history $\hist$ is an injective map
$\lin : \hist \to \mathbb{Q}{\ge 0}$ assigning each operation $o \in \hist$ a time $\lin(o) \in \rats_{\geq 0}$ 
such that $\invTimeOf{o} < \lin(o) < \resTimeOf{o}$.
In other words, $\lin$ defines a total order on operations consistent with the operation's invocation
and response obligations in the history.
For the stack history $\hist_{\sf stack}$ above, examples of linearizations of $\hist_{\sf stack}$ include
$\lin_1 = {o_1 \mapsto 2.5, o_2 \mapsto 3.5, o_3 \mapsto 6}$ and
$\lin_2 = {o_1 \mapsto 2.75, o_2 \mapsto 2.5, o_3\mapsto 6}$;
only $\lin_1$ agrees with the sequential specification of the (stack) ADT, formally described next.

\myparagraph{Sequential Specifications}
The sequential behavior of an ADT $\ADT$ plays a key role
in defining the linearizability specification for a concurrent implementation
and is, thus, also integral to our work; we formalize it here.
An abstract operation is a pair ${\tt a} = \tuple{m, v_{\argAttr}, v_{\retAttr}}$ 
where $m \in \methods_{\ADT}$ and 
$v_{\argAttr}, v_{\retAttr} \in \vals_\ADT$ (also written $\opr{m}{v_{\argAttr}}{v_{\retAttr}}$).
An abstract sequential history is a finite sequence $\tau = {\tt a}_1 \cdot {\tt a}_2 \cdots {\tt a}_n$
of abstract operations.
In other words, $\tau$ is a word over the (possibly infinite) alphabet 
$\alphabet_{\ADT} = \setpred{\opr{m}{v_{\retAttr}}{v_{\argAttr}}}{m \in \methods_{\ADT}, v_{\argAttr} \in \vals_{\ADT}, v_{\retAttr} \in \vals_{\ADT}}$
of abstract operations.
The sequential specification of an ADT (abstract data type) $\ADT$ with method set $\methods_{\ADT}$ is the prefix-closed set $\spec_{\ADT} \subseteq \alphabet^*_{\ADT}$ of all legal sequential behaviors of $\ADT$.
For example, the sequential specification $\QueueSpec$ of the \queueDS{} ADT contains
$\tau_1 = \opr{\enq}{\val{1}}{\vok} \cdot \opr{\enq}{\val{2}}{\vok} \cdot \opr{\deq}{}{\val{1}} \cdot \opr{\deq}{}{\val{2}}$
but not
$\tau_2 = \opr{\enq}{\val{1}}{\vok} \cdot \opr{\deq}{}{\val{2}} \cdot \opr{\enq}{\val{2}}{\vok}$.
Likewise, the sequential specification $\StackSpec$ of the \stackDS{} ADT contains
$\tau_3 = \opr{\pop}{}{\vfail} \cdot \opr{\push}{\val{3}}{\vok} \cdot \opr{\pop}{}{\val{3}}$
and $\tau_4 = \opr{\push}{\val{3}}{\vok} \cdot \opr{\pop}{}{\val{3}} \cdot \opr{\pop}{}{\vfail}$
but not $\tau_5 = \opr{\push}{\val{3}}{\vok} \cdot \opr{\pop}{}{\vfail} \cdot \opr{\pop}{}{\val{3}}$.

\myparagraph{Linearizability}
Let $\hist$ be a history of an ADT $\ADT$ and let $\lin$ be a linearization of $\hist$ (with $|\hist| = n$).
Let $o_1, o_2, \ldots, o_{n}$ be an enumeration of the operations of $\hist$
in accordance to $\lin$ (i.e., $\lin(o_{i}) < \lin(o_{i+1})$ for every $1 \leq i < n$).
The abstract sequential history induced by $\lin$
is the sequence
$\tau_\lin = \abs{o_1} \cdot \abs{o_2} \cdots \abs{o_n} \in \alphabet^*_{\ADT}$.
Here, for a given (concrete) operation $o = \tuple{id, p, X, m, v_{\argAttr}, v_{\retAttr}, t_{\invTimeAttr}, t_{\resTimeAttr}}$,
we use the notation $\abs{o} = \opr{m}{v_{\argAttr}}{v_{\retAttr}}$ to denote its abstract operation,
obtained by simply forgetting all components of $o$ except for the method identifier and the values of $o$.
A linearization $\lin$ is said to be \emph{legal} 
with respect to a sequential $\spec_{\ADT}$ if $\tau_\lin \in \spec_{\ADT}$; 
in this case we will write $\lin \in \spec_{\ADT}$.

\begin{definition}[Linearizable History]
A concurrent history $\hist$ is said to be linearizable with respect to the
sequential specification $\spec_{\ADT}$ of an ADT $\ADT$
if there is a linearization $\lin$ of $\hist$ such that $\lin \in \spec_{\ADT}$.
\end{definition}

The key focus of our work is the linearizability monitoring problem,
defined in the following, and its computational aspects.

\begin{problem}[Linearizability Monitoring]
Fix an ADT $\ADT$ with sequential specification $\spec_{\ADT}$.
Given a concurrent history $\hist$ as input,
decide whether $\hist$ is linearizable w.r.t.\ $\spec_{\ADT}$.
\end{problem}

\myparagraph{LTSs for sequential specifications}
In this work, we represent sequential specifications using labeled transition systems (LTSs).
Formally, for an ADT $\ADT$ with methods $\methods$ and values $\vals$, its LTS is a triple
${\sf LTS}_\ADT = (S_\ADT, s_0, \rightarrow)$, where
$S_\ADT$ is the set of states, $s_0 \in S_\ADT$ is the initial state, and
$\rightarrow \subseteq S_\ADT \times \alphabet_\ADT \times S_\ADT$
is the labeled transition relation.
We write $s \ltstrans{\opr{m}{v_{\argAttr}}{v_{\retAttr}}} s'$ to denote $(s, m(v), s') \in \rightarrow$.
A run of ${\sf LTS}_\ADT$ is a finite sequence
$\rho = s_0 \ltstrans{\tuple{m^{(1)}, v^{(1)}_{\argAttr}, v^{(1)}_{\retAttr}}} s_1 \cdots s_{k-1} \ltstrans{\tuple{m^{(k)}, v^{(k)}_{\argAttr}, v^{(k)}_{\retAttr}}} s_k$
of steps from the transition relation $\rightarrow$ starting from the initial state $s_0$.
The labeling of $\rho$, written $\lambda(\rho)$, is the sequence
$\tuple{m^{(1)}, v^{(1)}_{\argAttr}, v^{(1)}_{\retAttr}} \cdots \tuple{m^{(k)}, v^{(k)}_{\argAttr}, v^{(k)}_{\retAttr}}$ of abstract operations in $\rho$.
Finally, the sequential specification  of $\ADT$ can then
be described as the set $\spec_{\ADT} = \setpred{\lambda(\rho)}{\rho \text{ is a run of }{\sf LTS}_\ADT}$.
The LTS representation of the sequential specifications
for the ADTs $\stackDS{}$ and $\queueDS{}$ are shown in~\figref{stack-queue-lts}.
for reference.

\begin{figure}[t]
\specBox{
	\noindent\!\!\underline{\bf Stack}: 
	\quad $\methodsStack = \set{\push, \pop, \peek}$, $S_\stackDS = \vals^*$, $s_0 = \epsilon$ \\
	\vspace{-0.1in}
	\noindent
	\begin{align*}
		\begin{array}{cccc}
			\inferrule{s' = s \cdot v}{s \ltstrans{\opr{\push}{v}{\vok}} s'}
			&
			\inferrule{s = s' \cdot v}{s \ltstrans{\opr{\pop}{}{v}} s'}
			&
			\inferrule{\exists s'', s = s' = s'' \cdot v}{s \ltstrans{\opr{\peek}{}{v}} s'}
			&
			\inferrule{s = s' = \epsilon, m \in \set{\peek, \pop}}{s \ltstrans{\opr{m}{}{\vfail}} s'}
		\end{array}
	\end{align*}
	\hrule
    \vspace{0.1in}
	\noindent\!\!\underline{\bf Queue}: 
	\quad $\methodsQueue = \set{\enq, \deq, \peek}$, $S_\queueDS = \vals^*$, $s_0 = \epsilon$ \\
	\vspace{-0.1in}
	\noindent
	\begin{align*}
		\begin{array}{cccc}
			\inferrule{s' = v \cdot s}{s \ltstrans{\opr{\enq}{v}{\vok}} s'}
			&
			\inferrule{s = s' \cdot v}{s \ltstrans{\opr{\deq}{}{v}} s'}
			&
			\inferrule{\exists s'', s = s' = s'' \cdot v}{s \ltstrans{\opr{\peek}{}{v}} s'}
			&
			\inferrule{s = s' = \epsilon, m \in \set{\peek, \deq}}{s \ltstrans{\opr{m}{}{\vfail}} s'}
		\end{array}
	\end{align*}
}
\caption{Sequential specifications $\StackSpec{}$ and $\QueueSpec$ as LTSs.
$\vals$ is a generic value domain with $\set{\vok, \vfail} \cap \vals = \emptyset.$
\figlabel{stack-queue-lts}}
\end{figure}

\section{Linearizability Monitoring as Language Reachability}
\seclabel{lang-reachability}

The algorithms we propose in the paper rely on several key observations,
the most prominent being that linearizability monitoring can essentially be viewed
as a particular graph problem --- language reachability.
Language reachability asks if a given graph has a path labeled
from a fixed formal language, and has previously enjoyed
applications in static analyses~\cite{Conrado2025}.
In the following we recall the formal definition of this problem
and also show how the problem of linearizability monitoring of a concurrent
history reduces to this problem.

\myparagraph{Language reachability} 
Let $\alphabet$ be an alphabet.
Let $G = (N, E)$ be a $\alphabet$-labeled directed graph,
 where $N$ is finite set of nodes and the set $E$ contains edges
of the form $(u, a, v)$ where $u, v \in N, a \in \alphabet$.
For nodes $s, t \in N$, an $(s,t)$-walk of $G$ is a sequence 
$\pi = (u_0, a_1, u_1) \ldots (u_{n-1}, a_n, u_n)$ (with $n \geq 1$) 
such that $u_0 = s$, $u_n = t$, for each $0 \leq i < n$, $(u_i, a_{i+1}, u_{i+1}) \in E$.
Such a walk induces the word $\lbl{\pi} = a_1 \cdot a_2 \cdots a_n$.

\begin{problem}[Language reachability]
Fix a formal language $L \subseteq \alphabet^*$ over the alphabet $\alphabet$.
Given a $\alphabet$-labeled directed graph $G = (N, E)$ and distinguished
nodes $s, t \in N$, the $L$-reachability problem for $(G, s, t)$
asks to determine if there is an $(s, t)$-walk $\pi$ of $G$ whose induced
word satisfies $\lbl{\pi} \in L$.
\end{problem}

\myparagraph{Ideals and frontiers of a history}
Let $\hist$ be a concurrent history of an ADT $\ADT$.
We say that a sub-history $I \subseteq \hist$ is an \emph{ideal} of $\hist$
if the following holds: 
for every $o_1, o_2 \in \hist$ if $o_2 \in I$ and if $\resTimeOf{o_1} < \invTimeOf{o_2}$,
then $o_1 \in I$.
In other words, an ideal is simply a subset of $\hist$ that is downward closed
with respect to the partial order $\histpo{\hist}$.
We will use $\idealsOf{\hist}$ to denote the set of all ideals of $\hist$.
An ideal can be uniquely represented using its \emph{frontier}.
The frontier $f_I$ of an ideal $I$ is the set of its maximal elements (according to $\histpo{\hist}$),
i.e., $f_I = \setpred{o \in I}{\neg (\exists o' \in I, \resTimeOf{o} < \invTimeOf{o'})}$.
We will use $\frontiersOf{\hist} = \setpred{f_I}{I \in \idealsOf{\hist}}$ to represent the set
of all frontiers of $\hist$.
We remark that for two ideals $I_1, I_2 \in \idealsOf{\hist}$,
$f_{I_1} = f_{I_2}$ iff $I_1 = I_2$, and thus there is a bijection between
$\idealsOf{\hist}$ and $\frontiersOf{\hist}$.

\myparagraph{Frontier graph of a history}
Frontier graphs offer a systematic and succinct representation of the
set of all linearizations of a history, and have previously been used
to obtain algorithms for solving consistency problems
for applications like model checking and predictive testing~\cite{Gibbons1997,Shi2025MessagePassing,Mathur2020,tuncc2023optimal,Agarwal2021,chakraborty2024hard,cantin2005complexity,grains2024,farzan2026parametrizingreadsfromequivalencepredictive};
here we show that they also naturally model the linearizability problem.
Let $\hist$ be a history of an ADT $\ADT$.
The \emph{frontier graph} of a history $\hist$ is
a $\alphabet_{\ADT}$-labeled directed graph
 $\fg{\hist} = (\fgnodes{\hist}, \fgedges{\hist})$, where
\begin{itemize}
  \item the set of nodes is simply the set of all ideals, i.e., $\fgnodes{\hist} = \idealsOf{\hist}$, and
  \item the set of edges $\fgedges{\hist} = \setpred{(I_1, \abs{o}, I_2)}{I_1, I_2 \in \fgnodes{\hist}, I_2 = I_1 \uplus \set{o}}$ connects all pairs of ideals $(I_1, I_2)$ such that $I_2$ can be obtained by adding exactly
  one operation, say $o$, to $I_1$, and the label of the edge is precisely the abstract operation $\abs{o}$
  corresponding to $o$. 
\end{itemize}
We use $I^{\init}_{\hist} = \emptyset \in \fgnodes{\hist}$
and $I^{\final}_{\hist} = \hist \in \fgnodes{\hist}$ to denote
the source and destination nodes of the frontier graph respectively.
Some observations about the frontier graph $\fg{\hist}$ of a history $\hist$ are in order.
First, $\fg{\hist}$ is a directed acyclic graph since its edges only
go from an ideal of smaller size to an ideal of a larger size.
Second, its paths capture the set of all linearizations of the history:
\begin{proposition}
Let $\hist$ be a concurrent history.
For every linearization $\lin$ of $\hist$, there is a path $\pi$
in $\fg{\hist}$ starting in $I^{\init}_{\hist}$ and ending in $I^{\final}_{\hist}$
with $\lbl{\pi} = \tau_\lin$.
Further, for every path $\pi$
in $\fg{\hist}$ that starts in $I^{\init}_{\hist}$ and ends in $I^{\final}_{\hist}$, 
there is a linearization $\lin$
of $\hist$ such that $\lbl{\pi} = \tau_\lin$.
\end{proposition}

As a consequence, the linearizability monitoring problem can
be modeled as an instance of language reachability on this graph,
and has been the guiding light for many of the algorithms we propose in this work:

\begin{corollary}
\corlabel{linearizability-monitoring-as-language-reachability}
\protect{\circledsmall{C1}}
Let $\hist$ be a concurrent history of an ADT $\ADT$.
$\hist$ is linearizable against the sequential specification $\spec_{\ADT}$ 
iff $(\fg{\hist}, I^{\init}_{\hist}, I^{\final}_{\hist})$
satisfies $\spec_{\ADT}$-reachability.
\end{corollary}

Even though the above result reduces linearizability monitoring
to language reachability in the frontier graph, this in itself
is not sufficient for coming up with efficient algorithms.
First, language reachability can get intractable even on simple graphs.
Indeed, in cases when even the membership problem in the language $\spec_\ADT$ 
is intractable, the language-reachability problem (and also the linearizability monitoring problem)
is bound to be intractable.
Second, even when the language reachability problem can be solved in polynomial time,
the graph may itself be (exponentially) large.
Fortunately, we show that the size of graph is exponential only in the number
of processes of the history and not in the total number of operations in the history.
This follows because the partial order $\histpo{\hist}$ is an \emph{interval}
order and thus, the number of frontiers (and thus also ideals) is bounded above by an exponential in the width of the order. Finally, the degree of each node is also bounded by the width (number of processes).
Together we have the following:

\begin{lemma}
\lemlabel{frontier-graph-size}
Let $\hist$ be a history of ADT $\ADT$ and let $n = |\hist|$ and $k = |\procs_\hist|$.
Then, for the frontier graph $\fg{\hist} = (\fgnodes{\hist}, \fgedges{\hist})$, we have
$|\fgnodes{\hist}| \leq 1 + n\cdot 2^{k-1}$ and $|\fgedges{\hist}| \leq n\cdot k \cdot 2^{k-1}$.
\end{lemma}

In other words, the size of the frontier graph grows \emph{linearly} when the total
number of processes involved is constant, since the dependence on the this parameter
is purely multiplicative, paving the way for FPT algorithms for monitoring.
We remark that, this is in sharp contrast with the frontier graphs in problems
like consistency testing~\cite{Mathur2020} where the dependence on the number of processes/threads
is not purely multiplicative, but instead takes the form $O(n^{f(k)})$, where $f$
grows at least linearly in its argument.

\section{Linearizability Monitoring for Order-Oblivious Data Types}
\seclabel{oodt-lin}

In this section, we discuss the first class of data structures which
we call order-oblivious data types or OODTs, and their
generalization $\alpha$-OODTs (where $\alpha \in \nats_{> 0}$ is a constant positive natural number),
and show that 
the problem of linearizability monitoring for this class is 
in the complexity class FPT,
where we treat $k$ (the number of threads) as a parameter.


\subsection{Order-Oblivious Data Types}

Informally, \emph{order-oblivious data types} represent ADTs 
for which, the internal abstract state of the object, after having performed
a sequence $\tau$ of $n$ successful operations, can solely be determined
by the set of operations of $\tau$, 
independent of the precise order in which the operations of $\tau$ were performed.
Observe that many common data structures 
share this property --- the final value of a \counterDS{} object depends only on 
the net number of increments and decrements; 
the contents of a \setDS{} (or \multisetDS{}) depend only on which elements were inserted or removed; 
and the state of a \pqueueDS{} is determined solely by the 
multiset of currently stored items, 
independent of the order in which they were added/removed.
We provide a simple mathematical formalization for this class of ADTs in the following.

\myparagraph{State of an ADT}
While the internal state of the actual implementation may be very elaborate,
the abstract state can often be captured by the standard notion of
indistinguishability relation of formal languages~\cite{kozen2007automata}.
The indistinguishability relation $\sim_{\ADT}$ induced by the 
sequential specification $\spec_{\ADT}$ of an ADT $\ADT$ is 
the smallest congruence (under concatenation) 
satisfying the following property: 
for all sequences $\tau_1,\tau_2 \in \alphabet_{\ADT}^*$,
and for every continuation $\tau$,
if it holds that 
$\tau_1\cdot\tau \in \spec_{\ADT}$ iff $\tau_2\cdot\tau \in \spec_{\ADT}$,
then $\tau_1 \sim_{\ADT} \tau_2$.
Intuitively, two sequences are indistinguishable 
when they admit exactly the same set of legal extensions 
according to the sequential specification.
An abstract state of $\ADT$ can then be defined to be an equivalence class of $\sim_{\ADT}$.
We can now state the formal definition of OODTs:

\begin{definition}[Order-Oblivious Data Type]
\protect{\circledsmall{C2}}
An ADT~$\ADT$ with sequential specification~$\spec_{\ADT}$ 
is said to be \emph{order-oblivious} (or an OODT) if, for all sequences 
$\tau_1,\tau_2 \in \spec_{\ADT}$ {\em that are permutations of each other}, 
it holds that $\tau_1 \sim_{\ADT} \tau_2$.
\end{definition}
In other words, for two \underline{valid} abstract sequential histories 
(i.e., both belonging to $\spec_{\ADT}$) of an OODT $\ADT$, if
they only differ in the order of operations they perform, then the state
they end up in is the same.
As we remarked previously, order-obliviousness
is exhibited by several common data structures.
We present some OODTs using their LTS specifications in~\figref{oodt-specs}
and discuss them next.


\begin{figure}[t]
\specBox{
    \!\!\underline{\bf Set}: 
    \quad $\methodsSet = \set{\add, \remove, \contains}$, 
    $S_{\setDS} \subseteq \vals$, 
    $s_0 = \emptyset$ \\
    \vspace{-0.1in}
    \noindent
    \begin{align*}
      \begin{array}{cccccc}
        \inferrule{\begin{aligned}\begin{array}{c}v\notin s \\ s' = s \uplus \set{v}\end{array}\end{aligned}}{s \ltstrans{\opr{\add}{v}{\vok}} s'}
        &
        \!\!\inferrule{v \in s}{s \ltstrans{\opr{\add}{v}{\vfail}} s}
        &
        \!\!\inferrule{\begin{aligned}\begin{array}{c}v \in s \\ s' = s \setminus \set{v}\end{array}\end{aligned}}{s \ltstrans{\opr{\remove}{v}{\vok}} s'}
        &
        \!\!\inferrule{v \notin s}{s \ltstrans{\opr{\remove}{v}{\vfail}} s}
        &
        \!\!\inferrule{v \in s}{s \ltstrans{\opr{\contains}{v}{\vtrue}} s}
        &
        \!\!\inferrule{v \notin s}{s \ltstrans{\opr{\contains}{v}{\vfalse}} s}
      \end{array}
    \end{align*}
    \hrule
    \vspace{0.1in}
    \noindent\!\!\underline{\bf Multiset}: \
    $\methodsMultiset=\set{\add,\remove,\contains}$,\;
    $S_{\multisetDS} = [\vals\to\nats]$,\; $s_0 = \lambda v.\,0$
    \\[-0.1in]
    \begin{align*}
      \begin{array}{ccccc}
        \inferrule{s' = s[s(v)+1/v]}{s \ltstrans{\opr{\add}{v}{\vok}} s'} 
        &
        \inferrule{\begin{aligned}\begin{array}{c}s(v)>0 \\ s' = s[s(v)-1/v]\end{array}\end{aligned}}{s \ltstrans{\opr{\remove}{v}{\vok}} s'} 
        &
        \inferrule{\begin{aligned}\begin{array}{c}s(v)=0 \\ s' = s\end{array}\end{aligned}}{s \ltstrans{\opr{\remove}{v}{\vfail}} s'}
        &
        \inferrule{s(v)>0}{s \ltstrans{\opr{\contains}{v}{\vtrue}} s}
        &
        \inferrule{s(v)=0}{s \ltstrans{\opr{\contains}{v}{\vfalse}} s}
      \end{array}
    \end{align*}
    \hrule
    \vspace{0.1in}
    \!\!\underline{\bf Priority queue}: 
    \quad $\methodsPQueue = \set{\enq, \deq, \peek}$, 
    $S_\pqueueDS = \vals^*$, 
    $s_0 = \epsilon$ \\
    \vspace{-0.1in}
    \noindent
    \begin{align*}
      \begin{array}{ccccccc}
        \inferrule{s' = \mathsf{sort}_{\leq_{\vals}}(v \cdot s)}{s \ltstrans{\opr{\enq}{v}{\vok}} s'}
        &
        \;
        &
        \inferrule{\exists s'', s = s' = s'' \cdot v}{s \ltstrans{\opr{\peek}{}{v}} s'}
        &
        \;
        &
        \inferrule{s = s' \cdot v}{s \ltstrans{\opr{\deq}{}{v}} s'}
        &
        \;
        &
        \inferrule{s = s' = \epsilon, m \in \set{\peek, \deq}}{s \ltstrans{\opr{m}{}{\vfail}} s'}
      \end{array}
    \end{align*}
    \hrule
    \vspace{0.1in}
    \noindent\!\!\underline{\bf Counter}: \
    $\methodsCounter=\set{\inc,\dec,\get}$,\;
    $S_{\counterDS}=\nats$,\; $s_0=\val{0}$
    \\
    [-0.1in]
    \begin{align*}
      \begin{array}{ccccc}
        \inferrule{s' = s+1}{s \ltstrans{\opr{\inc}{}{\vok}} s'}
        &
        \inferrule{s>0 \;\wedge\; s' = s-1}{s \ltstrans{\opr{\dec}{}{\vok}} s'}
        &
        \inferrule{s=0 \;\wedge\; s' = s}{s \ltstrans{\opr{\dec}{}{\vfail}} s'}
        &
        \inferrule{s' = s = v \in \nats}{s \ltstrans{\opr{\get}{}{v}} s'}
      \end{array}
    \end{align*}
    \hrule
    \vspace{0.1in}
    \noindent\!\!\underline{\bf Fetch-and-add register}: \
    $\methodsFAA=\set{\mread,\faa}$,\;
    $S_{\faaDS}=\nats \uplus \set{\bot}$,\; $s_0=0$
    \\[-0.1in]
    \begin{align*}
    \begin{array}{cc}
    \inferrule{s = s' = v}{s \ltstrans{\opr{\mread}{v}{\vok}} s'}
    &
    \inferrule{s' = s + v}{s \ltstrans{\opr{\faa}{v}{s}} s'}
    \end{array}
    \end{align*}
    \hrule
    \vspace{0.1in}
    \noindent\!\!\underline{\bf Counting semaphore with $\ell$ permits)}~\cite{Dijkstra1968Cooperating}: \
    $\methodsSem=\set{\acq,\rel}$,\;
    $S_{\semDS}=\set{0,\dots,\ell}$,\;
    $s_0=\ell$
    \\[-0.1in]
    \begin{align*}
        \begin{array}{cccc}
            \inferrule{s>0,  s' = s-1}{s \ltstrans{\opr{\acq}{}{\vok}} s'}
            &
            \inferrule{s=0,  s' = s}{s \ltstrans{\opr{\acq}{}{\vfail}} s'}
            &
            \inferrule{s<\ell,  s' = s+1}{s \ltstrans{\opr{\rel}{}{\vok}} s'}
            &
            \inferrule{s=\ell, s' = s}{s \ltstrans{\opr{\rel}{}{\vfail}} s'}
        \end{array}
    \end{align*}
    \hrule
    \vspace{0.1in}
    \noindent\!\!\underline{\bf Compare-and-swap register}: \
    $\methodsCAS=\set{\mread,\cas}$,\;
    $S_{\casDS}=\vals \uplus \set{\bot} \setminus \set{\vok, \vfail}$,\; $s_0=\bot$
    \\[-0.1in]
    \begin{align*}
        \begin{array}{cccc}
            \inferrule{s = s' = v}{s \ltstrans{\opr{\mread}{}{v}} s'}
            &
            \inferrule{s = u \;\wedge\; s' = v}{s \ltstrans{\opr{\cas}{\tuple{u,v}}{\vok}} s'}
            &
            \inferrule{s = u \;\wedge\; s' = v}{s \ltstrans{\opr{\cas}{\tuple{u,v}}{\vfail}} s'}
            &
            \inferrule{s \ne u \;\wedge\; s' = s}{s \ltstrans{\opr{\cas}{\tuple{u,v}}{\vfail}} s'}
        \end{array}
    \end{align*}
    \hrule
    \vspace{0.1in}
    \noindent\!\!\underline{\bf Mutex}: \
    $\methodsMutex = \set{\acq,\rel}$,\;
    $S_{\mutexDS} = \procs \uplus \set{\bot}$,\;
    $s_0 = \bot$
    \\[-0.1in]
    \begin{align*}
        \begin{array}{cccc}
            \inferrule{s \in \set{\bot, p},  s' = p}{s \ltstrans{\opr{\acq}{p}{\vok}} s'}
            &
            \inferrule{s \notin \set{\bot, p}, s' = s}{s \ltstrans{\opr{\acq}{p}{\vfail}} s'}
            &
            \inferrule{s = p, s' = \bot}{s \ltstrans{\opr{\rel}{p}{\vok}} s'}
            &
            \inferrule{s \neq p, s' = s}
            {s \ltstrans{\opr{\rel}{p}{\vfail}} s'}
        \end{array}
    \end{align*}
    \hrule
    \vspace{0.1in}
    \noindent\!\!\underline{\bf Size-visible Stack}: 
    \quad $\methodsSizeStack = \set{\push, \pop, \peek}$, $S_\sizeStackDS = (\vals\setminus\set{\vfail})^*$, $s_0 = \epsilon$ \\
    \vspace{-0.1in}
    \noindent
    \begin{align*}
        \begin{array}{cccc}
            \inferrule{s' = s \cdot v, i = |s'|}{s \ltstrans{\opr{\push}{v}{\tuple{\vok, i}}} s'}
            &
            \inferrule{s = s' \cdot v, i = |s'|}{s \ltstrans{\opr{\pop}{}{\tuple{v, i}}} s'}
            &
            \inferrule{\exists s'', s = s' = s'' \cdot v}{s \ltstrans{\opr{\peek}{}{\tuple{v, i}}} s'}
            &
            \inferrule{s = s' = \epsilon, m \in \set{\peek, \pop}}{s \ltstrans{\opr{m}{}{\vfail}} s'}
        \end{array}
    \end{align*}
}
\caption{
Examples of OODTs. $\vals$ is a generic value domain satisfying $\set{\vok, \vfail} \cap \vals = \emptyset.$
\figlabel{oodt-specs}
}
\end{figure}

\myparagraph{Sets, multisets and priority queues}
\setDS{}s, \multisetDS{}s and \pqueueDS{}s are straightforwardly OODTs.
Indeed, for each of them, the collection of operations seen so far completely determines
the abstract state of the ADT.
In the case of a \setDS{}, the state can be determined
by the collection of those values $v$ for which the number of $\opr{\add}{v}{\vok}$ operations
is larger than the number of $\opr{\remove}{v}{\vok}$ operations.  
Likewise, for \multisetDS{}, the state can be determined by looking at the set of operations
seen so far, and collecting those values for which the number of additions is more
than the number of successful deletions.
For the case of \pqueueDS{} as well, it suffices to simply determine the
multiset of values that remain, which is determined solely by the set (and not sequence) of operations seen so far.
It must additionally be noted that, these data structures can further be extended
with following methods, without affecting their OODT status:
$\methodName{\sf max}$ and
$\methodName{\sf min}$ which respectively return the maximum and minimum
elements remaining in the underlying data structure,
or $\opr{\methodName{\sf countIf}}{P}{}$ that returns the number of elements
which satisfy the predicate $P$.

\myparagraph{Counters, fetch-and-add registers, compare-and-swap-registers and counting semaphores}
The abstract state of each of the ADTs \counterDS{}, \faaDS{}, \casDS and \semDS{}
can be represented by an internal counter, which can in turn be determined
by the multi-set of operations that modify this counter successfully.
Here we allow \casDS to fail spuriously as per existing implementations where
a weak exchange is used for better performance.
It must be noted that the return value in each of these ADTs ($\vok$ v/s $\vfail$) 
is crucial for us to be able to make this argument, and in absence of them, some of these ADTs
will cease to be OODTs; the same is also true for $\setDS{}$ and \multisetDS{} ADTs.

\myparagraph{Mutexes}
Let us argue why \mutexDS is an OODT. 
This is because, given a sequential \mutexDS{} history, we know that
at most one $\acq$ may be unmatched in the multiset of successful operations.
Further, the unmatched $\acq$ event also determines the unique owner of the mutex.
This completely determines the abstract state of the $\mutexDS{}$ object.

\myparagraph{Size-visible stacks}
As such, stacks are not OODTs since their state also tracks the precise order
in which values were pushed or popped.
An interesting variant is the \sizeStackDS{} ADT.
This data structure additionally makes visible, the size of the stack
in each operation.
The information of the multiset of such operations (independent of their relative order) 
is enough to recover the contents of the stack as follows.
The element at the $i^\text{th}$ position (from the bottom) of the stack
is the unique (if any) value $v$ for which the number of 
$\opr{\pop}{}{\tuple{v, i}}$ operations is smaller than the number of
$\opr{\push}{v}{\tuple{\vok, i}}$ operations.
For this reason,  \sizeStackDS{} is an OODT.

\begin{proposition}
\protect{\circledsmall{C2}}
Each of the ADTs \setDS{}, \multisetDS{}, \pqueueDS{}, \counterDS{}, \faaDS{}, \casDS{}, \semDS{},
\mutexDS{} and \sizeStackDS{} is an OODT.
\end{proposition}



\subsection{Generalizing OODTs}

Despite its generality, the OODT class still excludes
many ADTs, because their state retains 
limited dependence on the order of operations.
Take for example, the \registerDS{} ADT (\figref{alpha-oodt-specs}),
and two valid abstract sequential histories in $\RegisterSpec{}$:
$\tau_1 = \opr{\mwrite}{\val{2}}{\vok} \cdot \opr{\mwrite}{\val{1}}{\vok}$ and 
$\tau_2 = \opr{\mwrite}{\val{1}}{\vok} \cdot \opr{\mwrite}{\val{2}}{\vok}$.
Observe that, even though $\tau_1, \tau_2$ are both permutations,
we have $\tau_1 \not\sim_{\registerDS} \tau_2$, this is because
$\tau_1 \cdot \opr{\mread}{}{\val{1}} \in \RegisterSpec{}$
but $\tau_2 \cdot \opr{\mread}{}{\val{1}} \notin \RegisterSpec{}$,
making register not an OODT.
Nonetheless, the amount of ordering information retained in the state is extremely 
small --- its abstract state can be summarized by the value written most recently,
and further, can be effectively updated on-the-fly upon observing a future continuation.
This is precisely the motivation behind the generalization we consider here.
Many ADTs require a small finite 
amount of additional ``memory'' to record which of a few order 
configurations occurred, but beyond that, permutations of the same operations remain indistinguishable.
We formalize this with the help of a deterministic finite-state
transition system, as follows:

\begin{definition}[$\alpha$-Order-Oblivious Data Types]
\protect{\circledsmall{C2}}
Let $\alpha \in \nats_{\geq 0}$.
An ADT~$\ADT$ with sequential specification~$\spec_{\ADT}$ is an 
\emph{$\alpha$-Order-Oblivious Data Type} (\emph{$\alpha$-OODT}) 
if there exists a deterministic finite-state transition system
$\mathcal{A}_{\ADT} = (Q_\ADT, \alphabet_{\ADT}, \delta_\ADT, q_\ADT^0)$
with $\lvert Q_\ADT \rvert = \alpha$ such that, for all feasible sequences 
$\tau_1, \tau_2 \in \spec_{\ADT}$, 
if $\tau_1$ and $\tau_2$ are permutations of each other and
$\delta_\ADT(q^0_\ADT, \tau_1) = \delta_\ADT(q^0_\ADT, \tau_2)$, 
then $\tau_1 \sim_{\ADT} \tau_2$.
\end{definition}

In summary, $\alpha$-OODTs are those ADTs whose state is determined
by the set of operations (instead of the full sequence $\tau$) performed so far (as with OODTs), 
together with the control state of the $|\alpha|$ state transition 
system\footnote{Technically, in our setup,
the alphabet $\alphabet_\ADT$ of an $\ADT$ is rightfully allowed to be 
infinite because it mentions values.
In this case, as such, a finite transition system may not suffice to characterize
the class of ADTs like \registerDS where, as we show, the control state can track
the value in the register. 
While we omit this subtlety to keep our presentation crisp, it can nevertheless
also be formalized as follows --- an ADT $\ADT$ over a (finite or infinite) alphabet $\alphabet_\ADT$
falls into the generalized OODT class
if for every finite subset $\Gamma \subseteq_{\sf fin} \alphabet_\ADT$,
the sub-ADT $\spec_\ADT \cap \Gamma^*$ is an $\alpha$-OODT for some $\alpha \in \poly(|\Gamma|)$.
The insistence on $\poly(|\Gamma|)$ is largely to ensure that the running time of 
\algoref{oodt-lin} (\secref{oodt-algo}) is polynomial in the history's size.} 
that the (arbitrarily long) sequence $\tau$ results in.
$\alpha$-OODTs thus generalize OODTs and coincide with them when $\alpha=1$.
Also observe that this definition lends itself to the hierarchy: 
\[\text{OODTs} = 1\text{-OODTs} \subseteq 2\text{-OODTs} \subseteq 3\text{-OODTs} \subseteq  \cdots \]
In the following, we discuss a few representative 
data structures that exhibit finite dependence on operation order;
see \figref{alpha-oodt-specs} for their precise LTS specifications.


\begin{figure}[t!]
\specBox{
    \noindent\!\!\underline{\bf Register}: \
    $\methodsRegisters = \set{\mread, \mwrite}$,\;
    $S_{\registerDS} = \vals \uplus \set{\bot} \setminus \set{\vok}$,\;
    $s_0 = \bot$
    \\[-0.1in]
    \begin{align*}
        \begin{array}{cc}
            \inferrule{s = s' = v}{s \ltstrans{\opr{\mread}{}{v}} s'}
            &
            \inferrule{s' = v}{s \ltstrans{\opr{\mwrite}{v}{\vok}} s'}
        \end{array}
    \end{align*}
    \hrule
    \vspace{0.1in}
    \noindent\!\!\underline{\bf Map}: \
    $\methodsMap=\set{\mapPut, \remove, \get}$,\;
    $S_{\mapDS} = (\keys \rightharpoonup \vals)$, $s_0=\emptyset$
    \\[-0.1in]
    \begin{align*}
      \begin{array}{ccccc}
        \inferrule{s'=s[\key\mapsto v]}{s \ltstrans{\opr{\mapPut}{\tuple{\key,v}}{\vok}} s'}
        &
        \inferrule{\begin{aligned}\begin{array}{c}\key\in\dom(s)\\ s'=s\setminus\set{\key}\end{array}\end{aligned}}{s \ltstrans{\opr{\remove}{\key}{\vok}} s}
        &
        \inferrule{\key\notin\dom(s)}{s \ltstrans{\opr{\remove}{\key}{\vfail}} s}
        &
        \inferrule{\begin{aligned}\begin{array}{c}\key\in\dom(s)\\ v=s(\key)\end{array}\end{aligned}}{s \ltstrans{\opr{\get}{\key}{v}} s}
        &
        \inferrule{\key\notin\dom(s)}{s \ltstrans{\opr{\get}{\key}{\vfail}} s}
      \end{array}
    \end{align*}
}
\caption{Examples of $\alpha$-OODTs (with $\alpha > 1$). $\vals$ is a generic value domain satisfying $\set{\vok, \vfail} \cap \vals = \emptyset.$
\figlabel{alpha-oodt-specs}}
\end{figure}


\myparagraph{Registers and maps}
Consider again the \registerDS{} ADT which supports
$\opr{\mread}{}{v}$ and $\opr{\mwrite}{v}{\vok}$ operations.
As discussed before, this is not an OODT.
Nevertheless, it falls in the more general class above.
When the space of values $\vals$ that the register
can take on is finite, \registerDS{} is an $\alpha_{\registerDS}$-OODT for 
$\alpha_{\registerDS} = |\vals|+1$.
This follows because the state of the register can be 
captured using a simple finite transition system $\mathcal{A}_{\registerDS}$
that tracks the last value written in the sequence of operations seen so far.
More precisely, 
\[
\mathcal{A}_{\registerDS} = 
(Q_\registerDS, \alphabet_{\registerDS}, \delta_\registerDS, q^0_\registerDS),
\]
where $Q_\registerDS = \vals \uplus \set{\bot}$,
$q^0_\registerDS{} = \bot$,
$\delta_\registerDS(v, \opr{\mread}{}{v'}) = v$ and
$\delta_\registerDS(v, \opr{\mwrite}{v'}{\vok}) = v'$.
In the same spirit, a map over finitely many values $\vals$ on a single key $\key \in \keys$ is
an $\alpha_{\mapDS}$-OODT, where $\alpha_{\mapDS} = |\vals|+1$
\footnote{\mapDS{} supporting multiple keys is not truly an $\alpha$-OODT by our characterization;
doing so can risk $\alpha$ to be exponential in the number of keys
(and thus also in the history size).
Thankfully though, linearizability for \mapDS{} is local with respect to its keys.
Henceforth it suffices to assume here that the transition system only associates
with the state of the map projected to a key.}.

\begin{proposition}
\protect{\circledsmall{C2}}
Each of the ADTs \registerDS{} and \mapDS{} (with finite set of values and process IDs)
 is an $\alpha$-OODT, for some finite $\alpha$.
\end{proposition}




\subsection{Frontier Graph Algorithm for Linearizability Monitoring of $\alpha$-OODTs}
\seclabel{oodt-algo}

\myparagraph{Overview}
Recall from \corref{linearizability-monitoring-as-language-reachability} that
the task of linearizability monitoring of a history $\hist$ against a sequential specification
$\spec_{\ADT}$ can be solved by instead solving the $\spec_{\ADT}$-reachability
problem for the frontier graph $\fg{\hist}$.
Here, we show that the class of $\alpha$-OODTs can be monitored
for linearizability efficiently.
This follows because the nodes in the frontier
graph naturally expose part of the abstract state of such an ADT.
In particular, for the class OODTs (i.e., $1$-OODTs), each ideal (i.e., content of a given node)
is a strict refinement of the state of the ADT $\ADT$.
In turn, this means that a simple graph search
algorithm that visits each node in the frontier graph once, 
while determining membership in $\spec_{\ADT}$ at each node, can essentially solve the monitoring problem.
This insight can also systematically be generalized to the larger class
of $\alpha$-OODTs by additionally tracking the state of the automaton.

\begin{algorithm}
\caption{Linearizability monitoring for $\alpha$-OODT $\ADT$}
\algolabel{oodt-lin}
\myproc{$\OODTLin(\hist)$}
{
\Let $\fg{\hist} \gets (\fgnodes{\hist}, \fgedges{\hist})$ be the frontier graph of $\hist$\;
\Let $I^{\init}_{\hist} \gets \emptyset$, $I^{\final}_{\hist} \gets \hist$ \;
\Let $\mathcal{A}_{\ADT} = (Q_{\ADT}, \alphabet_{\ADT}, \delta_{\ADT}, q^0_{\ADT})$ be the finite-state transition system of $\ADT$ \;
\lForEach{$I \in \fgnodes{\hist}$, $q \in Q_{\ADT}$}
{
  $\repr_{I, q} \gets \bot$
}
$\repr_{I^{\init}_{\hist}, q^0_{\ADT}} \gets \varepsilon$ \;
\For{ideals $I \in \fgnodes{\hist}$ in increasing order of $|I|$}
{
  \ForEach{$q \in Q_{\ADT}$ such that $\repr_{I, q} \neq \bot$}
  {
    \For{$(I,\abs{o},I') \in \fgedges{\hist}$}
    {
      \Let $\tau' \gets \repr_{I, q} \cdot \abs{o}$, $q' \gets \delta_\ADT(q,\abs{o})$\;
      \lIf{$\tau' \in \spec_\ADT$ and $\repr_{I', q'} = \bot$}
      {
        $\repr_{I', q'} \gets \tau'$
      }
    }
  }
}
\Return $\big( \exists q \in Q_{\ADT} \cdot (\repr_{I^{\final}_{\hist}, q} \neq \bot)\big)$
}
\end{algorithm}

\myparagraph{Detailed description of the algorithm}
The frontier graph algorithm for linearizability monitoring
of an $\alpha$-OODT $\ADT$ is presented in \algoref{oodt-lin}.
The algorithm proceeds in a layered manner over the frontier graph
$\fg{\hist}$, processing ideals in increasing order of their size.
For each ideal $I$ and state $q \in Q_{\ADT}$, the algorithm maintains
a representative linearization $\tau \in \spec_{\ADT}$, denoted
$\repr_{I,q}$, such that $\tau$ labels some path from the initial ideal
$I^\init_\hist$ to $I$ and satisfies $\delta_\ADT^*(q^0_\ADT, \tau) = q$.
If no such linearization exists, we set $\repr_{I,q} = \bot$.
Initially, $\repr_{I^\init_\hist, q^0_\ADT} = \varepsilon$, and all other
entries are $\bot$.
The algorithm returns $\true$ iff there exists a state
$q \in Q_{\ADT}$ such that $\repr_{I^\final_\hist, q} \neq \bot$,
i.e., iff there exists a feasible linearization of $\hist$.

The algorithm processes ideals $I$ in increasing order of $|I|$.
For each $q \in Q_{\ADT}$ such that $\repr_{I,q} \neq \bot$,
it considers all outgoing edges $(I,\abs{o},I')$ in the frontier graph.
Let $\tau = \repr_{I,q}$. The algorithm attempts to extend $\tau$
by $\abs{o}$, forming $\tau' = \tau \cdot \abs{o}$.
If $\tau' \in \spec_{\ADT}$, then $\tau'$ is a valid linearization
for the successor ideal $I'$. Let $q' = \delta_\ADT(q,\abs{o})$.
If no representative has yet been recorded for $(I',q')$,
the algorithm sets $\repr_{I',q'} \gets \tau'$.
Thus, for each ideal $I$, the algorithm maintains at most one
representative linearization for every state $q \in Q_{\ADT}$.
The correctness of this procedure follows from the definition of
$\alpha$-OODTs: for a fixed ideal $I$, all feasible linearizations
of $I$ are permutations of one another, and hence any two
linearizations that reach the same state $q$ are equivalent
for all future extensions. Therefore, it suffices to retain
a single representative per state.

The time complexity of the algorithm can be determined by observing
that for each ideal $I$, at most $\alpha$ representatives are maintained,
one per state in $Q_{\ADT}$. Each representative is propagated along
every outgoing edge $(I,\abs{o},I')$ at most once, and each propagation
requires a membership check in $\spec_{\ADT}$. Assuming that the function
${\sf mem}_{\ADT}: \nats \to \nats$ bounds the running time for membership
in $\spec_\ADT$, we obtain the following:

\begin{theorem}
  \thmlabel{alpha-OODT-FPT}
  \protect{\circledsmall{C2}}
  Let $\ADT$ be an $\alpha$-OODT.
  Given a history $\hist$ with $n$ operations and $k$ processes, 
  \algoref{oodt-lin} runs in time $O(\alpha{}nk2^k \cdot {\sf mem}_{\ADT}(n))$
  and returns true iff $\hist$ is linearizable.
\end{theorem}

For many data structures in \figref{oodt-specs} and \figref{alpha-oodt-specs},
including \registerDS{} and \pqueueDS{}~\cite{Finkler1999} (interestingly), 
there is a linear time procedure for solving the membership of their specifications. 
This automatically gives us a $O(\alpha{}k2^k \cdot n^2)$ time for their linearizability monitoring 
problem.
Nonetheless, there is obvious room for improvements.
We can see that instead of constructing and verifying an entire linearization $\tau$ 
at each node of the frontier graph, we can perform the membership check
incrementally,  by maintaining a simulated state of the object.
This would, however, often require the operation to be 
backtrackable in an equi-efficient fashion.
For example, since a \pqueueDS{} can be simulated using a self-balancing binary search tree,
where each updates are (also backtrackable in) $O(\log{n})$ time,
we have that the linearizability monitoring for \pqueueDS can be solved in $O(k2^k\cdot n\log{n})$ time.
Similarly, linearizability monitoring for hash-based (multi)sets can be solved in $O(nk2^k)$ time.


\section{Linearizability Monitoring for Context-Free Data Types}
\seclabel{cfl-adt}

In this section, we consider the \stackDS ADT and study its linearizability problem.
We remark that, even with a single value, stacks do not fall into the
OODT class, or for that matter even in  $\alpha$-OODT, no matter what $\alpha$ we pick.
This is because, in order to check the membership $\tau \in \stackDS$, one
requires space that depends upon $|\tau|$.
Nonetheless, we show that linearizability monitoring for $\stackDS$
can be solved in time that is proportional to $O(\poly(n) \cdot c^{k})$
for histories of size $n$ and $k$ processes, i.e., this problem still remains in FPT
where the number of processes is the parameter ($c$ is a fixed constant).
Our algorithm essentially solves an instance of the
CFL-reachability problem for the frontier graph, and for this reason
can be generalized to a wider class of context-free ADTs, which we define next.

\subsection{Context-Free Data Types and their Linearizability}

To cater to the subtlety that ADTs are canonically defined over
a possibly infinite alphabet, while most common presentations (and algorithms thereof)
for CFLs are over a finite alphabet, we define
this class using finite projections:

\begin{definition}[Context-free data type]
\protect{\circledsmall{C3}}
An ADT $\ADT$ with sequential specification $\spec_\ADT$ (over alphabet $\alphabet_\ADT$) is a
context-free data type (CFDT, for short), if for every finite subset $\Gamma \subseteq_{\sf fin} \alphabet_\ADT$,
the language $\spec_\ADT \cap \Gamma^*$ is a context free language.
\end{definition}

Given that the CFL-reachability problem, for a grammar of size $\gamma$
over a graph with $N$ vertices can be solved in time 
$O(\gamma \cdot N^3)$~\cite{CFLReachabilityPavlogianis2023}, we have the following
FPT result:

\begin{theorem}
\thmlabel{CFL-ADT-FPT}
\protect{\circledsmall{C3}}
Let $\ADT$ be a CFDT with sequential specification $\spec_\ADT$ (over alphabet $\alphabet_\ADT$).
The linearizability monitoring problem for a given concurrent history $\hist$
with $n$ operations and $k$ processes can be solved in time $O(g_\hist \cdot n^3 \cdot 2^{3k})$,
where $g_\hist$ is the size of the smallest grammar that describes the CFL $\spec_\ADT \cap \Gamma_\hist^*$,
where $\Gamma_H = \set{\abs{o}}_{o \in \hist}$.
\end{theorem}

For completeness and for facilitating our discussion in later sections,
we provide in \algoref{stack-lin} the frontier graph algorithm for
linearizability monitoring for a context-free ADT $\ADT{}$.

\begin{algorithm}[t]
\caption{Linearizability monitoring for Context-free data type $\ADT$}
\algolabel{stack-lin}
\myproc{$\StackLin(\hist)$}{
\Let $\fg{\hist} \gets (\fgnodes{\hist}, \fgedges{\hist})$ be the frontier graph of $\hist$\;
\Let $\mathcal{C}_{\ADT} = (\nonTerminals{\ADT}, \alphabet_{\ADT, \hist}, \prodRules{\ADT}, S_{\ADT})$ be the CFG for $\spec_{\ADT} \cap \alphabet_{\ADT, \hist}^*$, with $\alphabet_{\ADT, \hist} = \set{\abs{o}}_{o \in \hist}$\;
\Let $\nullableNonTerminals{\ADT} \gets \setpred{A \in \nonTerminals{\ADT}}{A \Rightarrow_{\ADT}^{*} \epsilon}$ \tcp*[h]{nullable non-terminals} \;

\lForEach{$(I, I') \in \fgnodes{\hist} \times \fgnodes{\hist}$}{
   $M(I, I') \gets \emptyset$
}

\tcp{Seed diagonal entries using nullable non-terminals}
\lForEach{$I \in \fgnodes{\hist}$}{
  $M(I,I) \gets \nullableNonTerminals{\ADT}$
}

\tcp{Seed entries corresponding to single frontier-graph edges}
\lForEach{$(I_1, \abs{o}, I_2) \in \fgedges{\hist}$ and $A \prodRules{\ADT} \abs{o}$}{
    $M(I_1, I_2) \gets M(I_1, I_2) \cup \set{A}$
}

\tcp{Dynamic programming over pairs $(I_1,I_2)$ in increasing order of $|I_2|-|I_1|$}
\ForEach{$(I_1, I_2)$ \WHERE $I_1 \subseteq I_2$, in increasing order of difference $|I_2| - |I_1|$}{
  \ForEach{$I_3$ \WHERE $I_1 \subseteq I_3 \subseteq I_2$}{
    $M(I_1, I_2) \gets M(I_1, I_2) \cup 
      \setpred{A}{
        A \rightarrow_\ADT B_1B_2,\;
        B_1 \in M(I_1, I_3),\;
        B_2 \in M(I_3, I_2)
      }$
      \linelabel{cfg-production}
  }
}

\Return $S_{\ADT} \in M(I^{\init}_{\hist}, I^{\final}_{\hist})$
}
\end{algorithm}

\myparagraph{Detailed description of the algorithm}
We work with the grammar $\mathcal{C}_{\ADT} = (\nonTerminals{\ADT}, \alphabet_{\ADT, \hist}, \prodRules{\ADT}, S_{\ADT})$ for $\ADT$, restricted to operations in $\hist$.
We assume $\mathcal{C}_{\ADT}$ is in Chomsky Normal Form, allowing $\epsilon$-derivations.
The algorithm builds a table $M$ where $M(I_1, I_2)$ contains all non-terminals that generate a labeled path between ideals $I_1$ and $I_2$ in the frontier graph.
Thus, $\hist$ is linearizable iff $S_{\ADT} \in M(I^{\init}_{\hist}, I^{\final}_{\hist})$.
Initialization proceeds as follows.
For every ideal $I$, we set $M(I,I)$ to the set of non-terminals deriving $\epsilon$.
For every edge $(I_1, \abs{o}, I_2)$, we add to $M(I_1, I_2)$ all non-terminals generating $\abs{o}$.
We then fill $M$ in increasing order of $|I_2| - |I_1|$.
For each pair $(I_1, I_2)$ and each $I_1 \subseteq I_3 \subseteq I_2$,
we add $A$ to $M(I_1, I_2)$ if $A \rightarrow B_1 B_2$ and
$B_1 \in M(I_1, I_3)$, $B_2 \in M(I_3, I_2)$.
Allowing $I_3 = I_1$ or $I_2$ accounts for $\epsilon$-derivations.
This is analogous to CYK parsing over the acyclic frontier graph,
where any two paths between a pair of ideals are permutations of each other.

\subsection{Optimizing linearizability monitoring for \stackDSUpper{}}

\newcommand{\reducedStSp}{\widehat{\spec}_{\stackDS{}}}

Here, we zoom into the linearizability monitoring
for the \stackDS ADT as defined in \figref{stack-queue-lts}, and discuss fine-grained optimizations
that result from specific properties of its sequential specification $\StackSpec{}$.
Before we proceed, we outline the context free grammar for stacks,
since it will be useful to present our optimizations.
For simplicity, we first only cater to the subset $\reducedStSp \subseteq \StackSpec{}$ 
that only contains sequences that are non-empty, well-matched (i.e., every $\push$ has a corresponding $\pop$)
and do not contain failed operations $\opr{\peek}{}{\vfail}$ and $\opr{\pop}{}{\vfail}$.
We address the full specification later.

\myparagraph{CFG for $\reducedStSp{}$}
Let $\vals$ be a finite subset of values, and let 
 $\alphabet^\vals_{\stackDS} = \setpred{\opr{\push}{v}{\vok}, \opr{\pop}{}{v}, \opr{\peek}{}{v}}{v \in \vals}$.
The CFG for stacks over values is the tuple
$\mathcal{C}^\vals_\stackDS = (\nonTerminals{\stackDS}, \alphabet^\vals_\stackDS, \prodRules{\stackDS}, S_\stackDS)$, where
$\nonTerminals{\stackDS} = \set{S_\varepsilon} \cup \setpred{S_v, S_{\opr{\push}{}{v}}, S_{\opr{\peek}{}{v}}}{v \in \vals}$, is the set of non-terminals, and
$S_\stackDS = S_\varepsilon$ is the starting symbol.
The production rules of $\mathcal{C}^\vals_\stackDS$ are the following (for each $v \in \vals$):

\begin{align*}
\begin{array}{rcl}
S_\varepsilon & \rightarrow_\stackDS & S_{\opr{\push}{}{v}} \, S_v \\
S_\varepsilon & \rightarrow_\stackDS & S_\varepsilon \, S_\varepsilon \\
S_v           & \rightarrow_\stackDS & S_{\opr{\peek}{}{v}} \, S_v \\
S_v           & \rightarrow_\stackDS & S_\varepsilon \, S_v
\end{array}
\qquad
\begin{array}{rcl}
S_v           & \rightarrow_\stackDS & \opr{\pop}{}{v} \\
S_{\opr{\push}{}{v}} & \rightarrow_\stackDS & \opr{\push}{v}{\vok} \\
S_{\opr{\peek}{}{v}} & \rightarrow_\stackDS & \opr{\peek}{}{v}
\end{array}
\end{align*}

Intuitively, the non-terminal $S_v$ (for some $v \in \vals$)
generates those sequences $\tau$ that are legal sequences of operations, assuming 
the internal stack state contains the value $v$ at the top,
i.e., $\opr{\push}{}{v} \cdot \tau \in \reducedStSp{}$.
Likewise, the non-terminal $S_\varepsilon$ generates sequences in $\StackSpec{}$,
i.e., those sequences that are legal stack sequences starting from any stack
configuration (including the empty stack).
The formal statement about the correctness of the above grammar is as follows:
\begin{proposition}
Given a stack history $\hist$, 
we have $\lang{\mathcal{C}^\vals_\stackDS} = \reducedStSp \cap \Gamma_\hist^*$;
here $\Gamma_H = \set{\abs{o}}_{o \in \hist}$, $\vals = \vals_\hist$ and
and $\lang{\mathcal{C}^\vals_\stackDS}$ is the set of sequences accepted
by the grammar $\mathcal{C}^\vals_\stackDS$.
\end{proposition}

Here we can see that for a given stack history $\hist$, $g_\hist = O(|\vals_\hist|)$.
We hence have that $\reducedStSp{}$ reachability for $(\fg{\hist}, I^{\init}_{\hist}, I^{\final}_{\hist})$
can be solved in $O(|\vals_\hist| \cdot |\fgnodes{\hist}|^3)$.

\myparagraph{Removing the $|\vals_\hist|$ factor}
The generic CFDT monitoring algorithm (\algoref{stack-lin})
runs in time $O(|\vals_\hist| \cdot |\fgnodes{\hist}|)$,
where the $|\vals_\hist|$ factor arises at \lineref{cfg-production}
when computing
\[
\setpred{A}{A \rightarrow_\ADT B_1B_2,\; B_1 \in M(I_1,I_3),\; B_2 \in M(I_3,I_2)},
\]
by iterating over $O(|\vals_\hist|)$ non-terminals.
Nonetheless, the grammar $\mathcal{C}^\vals_\stackDS$ is functional:
each $\pi \in \Gamma_{\hist}^*$ is generated by at most one non-terminal.
Thus $|M(I_1,I_2)| \leq 1$ for all ideals $I_1,I_2$.
Moreover, the unique non-terminal (if any) is determined by
the $\push/\pop$ balance in $I_2 \setminus I_1$
and can be precomputed.
Hence the above set is computable in $O(1)$ time,
removing the $|\vals_\hist|$ factor.
\begin{lemma}
  Let $\hist$ be a concurrent \stackDS history. $\reducedStSp$ reachability for $(\fg{\hist}, I^{\init}_{\hist}, I^{\final}_{\hist})$ can be solved in $O(|\fgnodes{\hist}|^3)$.
\end{lemma}


\myparagraph{Sub-Cubic Matrix Multiplication}
It is known that the problem of parsing context-free grammar is solvable via algorithms that runs in sub-cubic time complexity,
thanks to the reduction to boolean matrix multiplication as shown by Valiant~\cite{Valiant1974}.
It has also been shown that CFL-reachability enjoys similar reduction~\cite{koutris2023fine}.
However, a direct application of the above on \stackDS yields a time complexity
poorer than what we previously established.
We show the reduction naturally extends to our case of frontier graph reachability of $\reducedStSp$,
therefore also arriving at a sub-cubic solution.
Note that the derived algorithm serves mostly theoretical interests
and is impractical for implementation due to a large constant factor overhead.
This gives us the following.

\begin{lemma}
  Let $\hist$ be a concurrent \stackDS history.
  $\reducedStSp{}$ reachability for $(\fg{\hist}, I^{\init}_{\hist}, I^{\final}_{\hist})$
  can be solved in $O(|\procs_\hist| \cdot |\fgnodes{\hist}|^\omega)$.
\end{lemma}

\newcommand{\mfunc}{\mathcal{F}}
\newcommand{\rev}{\mathsf{rev}}
\newcommand{\revMethodOf}[1]{\methodAttr^{\rev}(#1)}
\newcommand{\revInvTimeOf}[2]{\invTimeAttr^\rev_{#1}(#2)}
\newcommand{\revResTimeOf}[1]{\resTimeAttr^\rev_{#1}(#1)}

\myparagraph{Handling empty histories, failed operations, and non-well-matchedness}
Let us now address the linearizability monitoring problem for the full specification
$\StackSpec{}$ (as against the reduced specification $\reducedStSp{} \subsetneq \StackSpec{}$ so far in this section).
First, $\epsilon \in \StackSpec{}$ and thus an empty history is always linearizable.

We now turn to failed operations (i.e. those of the form
 $\opr{\peek}{}{\vfail}$ or $\opr{\pop}{}{\vfail}$).
Informally, these operations can be accounted for assuming that
they access a special fresh value $\bot$ that always remains at the bottom of the stack.
Formally, consider the homomorphism on operations given by:
\[
    \mfunc_\bot(o) = 
    \begin{cases}
        o[\methodOf{o} \mapsto \peek, \returnOf{o} \mapsto \bot] & \text{ if } \returnOf{o} = \vfail\\
        \mfunc_\bot(o) = o &  \text{ otherwise}
    \end{cases}
\]
and let $\mfunc_\bot(\hist) = \setpred{\mfunc_\bot(o)}{o\in\hist}$
be the image of the history under the transformation $\mfunc_\bot$.
Now, consider the new history $\hist_\bot = \mfunc_\bot(\hist) \cup \set{\tuple{\_, \push, \bot, \vok, t_{\min{}} - 2, t_{\min{}} - 1}}$ obtained by adding a dummy push operation for the fresh value $\bot$ in the beginning of the history.
Observe that $\hist_\bot$ has no fail operations,
can be computed in $O(|\hist|)$ time and most importantly is equi-linearizable with $\hist$ (see below),
allowing us to assume that there are no fail operations w.l.o.g:
\begin{proposition}
  $\hist$ is linearizable if and only if $\hist_\bot$ is linearizable.
\end{proposition}

We now address the full stack specification, and in particular,
the case when histories do not have equal number of $\push$ and $\pop$
operations (i.e., they cannot be linearized to well-matched stack sequences).
To account for such histories, we effectively make these histories well-matched
by appending a ``mirrored'' copy of the history to itself.
Formally, for an operation $o \in \hist$, consider its dual operation:
\[
    \rev(o) = \tuple{\procOf{o}, \revMethodOf{\methodOf{o}}, \argOf{o}, \returnOf{o}, \revResTimeOf{\hist}{o}, \revInvTimeOf{\hist}{o}}
\]
where $\revMethodOf{\push} = \pop$, $\revMethodOf{\pop} = \push$, $\revMethodOf{\peek} = \peek$,
$\revResTimeOf{\hist}{o} = 2\cdot \mu_\hist + 1 - \invTimeOf{o}$, 
$\revInvTimeOf{\hist}{o} = 2\cdot \mu_\hist + 1 - \resTimeOf{o}$,
with $\mu_\hist = \max\setpred{\resTimeOf{o}}{o \in \hist}$.
Now let $\hist^\rev = \setpred{\rev(o)}{o \in \hist}$ and consider the history
$\hist_\match = \hist \cup \hist^\rev$.
$\hist_\match$ is well matched, can be constructed in $O(\hist)$ time
and most importantly:
\begin{proposition}
 $\hist$ is linearizable if and only if $\hist_\match$ is linearizable.
\end{proposition}

Together, we have the following:

\begin{theorem}
    \thmlabel{stack-FPT}
    \protect{\circledsmall{C3}}
  Given a \stackDS{} history $\hist$ with $n$ operations and $k$ processes,
  linearizability monitoring can be solved in $O(\min\set{2^{3k}\cdot n^3, k^\omega{}2^{\omega{}k} \cdot n^\omega{}})$ time.
\end{theorem}

\section{Linearizability Monitoring for $\queueDSUpper$}
\seclabel{queue-adt}

The $\queueDS{}$ ADT presents a qualitatively different challenge:
its specification is neither order-oblivious nor context-free.
At a high level, $\QueueSpec{}$ is context-sensitive, for which even
membership—and hence language reachability over general graphs—is
undecidable (\secref{queue-lang-reachability}).
Restricting to DAGs does not help: queue reachability remains $\np$-hard.
Thus, tractability does not follow from known results.
Nonetheless, we obtain an FPT algorithm by exploiting the structure of
frontier graphs. The key idea is to reformulate the specification as a
relational rewrite system (\secref{queue-rewrite}) and reduce reachability
to a fixpoint computation.


\subsection{A Rewrite system for $\QueueSpec{}$}
\seclabel{queue-rewrite}


The original formulation of $\QueueSpec{}$ via ${\sf LTS}_{\queueDS{}}$
explicitly tracks the entire queue content and is not amenable to efficient
reachability. We therefore reformulate the specification using relations that
track only limited information about the state.
The key observation is that only the current front of the queue needs to be
tracked precisely; all other values can be handled abstractly. Accordingly,
we view recognition as consuming operations while maintaining a processed
prefix, a remaining suffix, and a distinguished front value (or $\qbot$ if none
is present). The rules below update this front value and allow other operations
to commute past it.

\myparagraph{Relational rewrite system for $\QueueSpec{}$}
To formalize this, we define relations
$R_{v_{\qbot}} \subseteq \alphabet_{\queueDS}^* \times \alphabet_{\queueDS}^*$
for each $v_{\qbot} \in \vals \uplus \set{\qbot}$.
These are the least family $\set{R_{v_{\qbot}}}_{v_{\qbot} \in \vals \uplus \set{\qbot}}$
closed under the following rules;
here $\alpha,\beta \in \alphabet_{\queueDS}^*$, 
$v\in \vals$, $v_{\qbot} \in \vals \uplus \set{\qbot}$ and $v_{\vfail} \in \vals \cup \set{\vfail} $):

{\small
\setlength{\arraycolsep}{2pt}
\begin{align*}
\begin{array}{rcl@{\quad}rcl}
\multicolumn{6}{c}{(\epsilon,\epsilon) \in R_{\qbot}}\\
(\alpha,\epsilon) \in R_{\qbot}
&\!\Rightarrow\!&
(\alpha, \opr{\peek}{}{\vfail}) \in R_{\qbot}
&
(\alpha,\epsilon) \in R_{\qbot}
&\!\Rightarrow\!&
(\alpha, \opr{\deq}{}{\vfail}) \in R_{\qbot} 
\\
(\alpha,\beta) \in R_v
&\!\Rightarrow\!&
(\alpha, \beta{\cdot}\opr{\peek}{}{v}) \in R_v
&
(\alpha,\beta) \in R_v
&\!\Rightarrow\!&
(\alpha, \beta{\cdot}\opr{\deq}{}{v}) \in R_{\qbot}
\\
(\alpha,\beta) \in R_{v_{\qbot}}
&\!\Rightarrow\!&
(\alpha, \beta{\cdot}\opr{\enq}{v}{\vok}) \in R_{v_{\qbot}}
&
(\alpha, \opr{\enq}{v}{\vok}{\cdot}\beta) \in R_{\qbot}
&\!\Rightarrow\!&
(\alpha{\cdot}\opr{\enq}{v}{\vok}, \beta) \in R_v
\\
(\alpha, \opr{\peek}{}{v_{\vfail}}{\cdot}\beta) \in R_{v_{\qbot}}
&\!\Rightarrow\!&
(\alpha{\cdot}\opr{\peek}{}{v_{\vfail}}, \beta) \in R_{v_{\qbot}}
&
(\alpha, \opr{\deq}{}{v_{\vfail}}{\cdot}\beta) \in R_{v_{\qbot}}
&\!\Rightarrow\!&
(\alpha{\cdot}\opr{\deq}{}{v_{\vfail}}, \beta) \in R_{v_{\qbot}}
\end{array}
\end{align*}
}

Informally, when $(\alpha,\beta) \in R_{v_{\qbot}}$, the sequence
$\tau = \alpha \cdot \beta$ is a valid queue execution.
Moreover, among the enqueue operations occurring in $\alpha$,
all of them are matched by dequeue operations in $\beta$
when $v_{\qbot}=\qbot$,
and all but one are matched when $v_{\qbot}=v$.
In the latter case, the unique unmatched enqueue in $\alpha$
has value $v$.
Intuitively, this corresponds to placing a pointer in $\tau$
immediately after the distinguished enqueue operation in $\alpha$.
We can now show the following correspondence:

\begin{restatable}{lemma}{QueueRewriteSystem}
\lemlabel{queue-rewrite-correctness}
$\QueueSpec{} = \bigcup\limits_{v_{\qbot} \in \vals \cup \set{\qbot}} \setpred{\alpha \cdot \beta}{(\alpha, \beta) \in R_{v_{\qbot}}}$
\end{restatable}

\subsection{Linearizability Monitoring for $\queueDS{}$}
\seclabel{queue-lang-reachability}

\myparagraph{$\QueueSpec{}$-reachability on arbitrary (cyclic or acyclic) graphs}
Before turning to frontier graphs, let us recall that queue-language reachability
is already hard in much more general graph classes.
For arbitrary directed graphs, the problem is undecidable.
Indeed, state reachability for queue automata is undecidable, and queue automata
embed directly into our framework: one views each automaton state as a graph node,
and each transition labeled by an enqueue or dequeue operation as a graph edge.
Thus, asking whether a designated automaton state is reachable is a special case
of asking whether there exists a path in the graph whose label belongs to
$\QueueSpec{}$.
On the other hand, if the graph is restricted to be acyclic, the problem becomes
decidable but remains intractable: \cite{schnoebelen2021} shows $\np$-hardness for
reachability in acyclic queue automata, via a reduction that already yields DAGs
of bounded width. Since acyclic queue automata again form a special case of our
setting, it follows that queue-language reachability is $\np$-hard even on bounded-width DAGs.

\myparagraph{Frontier graphs}
Frontier graphs arise from a single partially ordered set of operations:
nodes are ideals and edges add a single operation. Hence, every path
corresponds to a linearization of the same set of operations, and paths
differ only in the ordering of independent operations.
This allows reasoning about paths to be reduced to reasoning about sets
of operations (ideals).
More importantly, explicit representation of nodes as ideals, also allows us to
exploit observations about when linearizations can be reordered to a normal form
by commuting operations against other independent operations.
 We exploit this by lifting the rewrite system
from sequences to ideals, yielding a notion of \emph{queue ideals reachability}.

\begin{definition}[Queue relations over ideals]
\deflabel{queue-ideal-rel}
Let $\hist$ be a concurrent $\queueDS{}$ history.
We define a family of binary relations
$\set{\IdlRel{v_{\qbot}}}_{v_{\qbot} \in \vals_\hist \cup \set{\qbot}}$ (where $\IdlRel{v_{\qbot}} \subseteq \fgnodes{\hist} \times \fgnodes{\hist}$)
as the least family of relations satisfying the following:
for each $I_1, I_2, I_3\in \fgnodes{\hist}$ and
$v \in \vals_\hist$,
$v_{\qbot} \in \vals_\hist \cup \set{\qbot}$,
$v_{\vfail} \in \vals_\hist \cup \set{\vfail}$:
\begin{enumerate}
  \item
  $(I^{\init}_{\hist}, I^{\init}_{\hist}) \in \IdlRel{\qbot}$.

    \item
  $(I_1, I_1) \in \IdlRel{\qbot} \land (I_1, \opr{\deq}{}{\vfail}, I_2) \in E_\hist \implies (I_1, I_2) \in \IdlRel{\qbot}$.

    \item
  $(I_1, I_1) \in \IdlRel{\qbot} \land (I_1, \opr{\peek}{}{\vfail}, I_2) \in E_\hist \implies (I_1, I_2) \in \IdlRel{\qbot}$.

    \item
  $(I_1, I_2) \in \IdlRel{v} \land I_1 \subseteq I_2 \subseteq I_3 \land (I_2, \opr{\peek}{}{v}, I_3) \in E_\hist \implies (I_1, I_3) \in \IdlRel{v}$.

  \item
  $(I_1, I_2) \in \IdlRel{v} \land I_1 \subseteq I_2 \subseteq I_3 \land (I_2, \opr{\deq}{}{v}, I_3) \in E_\hist \implies (I_1, I_3) \in \IdlRel{\qbot}$.

   \item
  $(I_1, I_2) \in \IdlRel{v_{\qbot}}  \land I_1 \subseteq I_2 \subseteq I_3 \land (I_2, \opr{\enq}{v}{\vok}, I_3) \in E_\hist \implies (I_1, I_3) \in \IdlRel{v_{\qbot}}$.

  \item
  $(I_1, I_3) \in \IdlRel{\qbot} \land I_1 \subseteq I_2 \subseteq I_3 \land (I_1, \opr{\enq}{v}{\vok}, I_2) \in E_\hist \implies (I_2, I_3) \in \IdlRel{v}$.

  \item
  $(I_1, I_3) \in \IdlRel{v_{\qbot}} \land I_1 \subseteq I_2 \subseteq I_3 \land (I_1, \opr{\peek}{}{v_{\vfail}}, I_2) \in E_\hist \implies (I_2, I_3) \in \IdlRel{v_{\qbot}}$.

  \item
  $(I_1, I_3) \in \IdlRel{v_{\qbot}}  \land I_1 \subseteq I_2 \subseteq I_3 \land (I_1, \opr{\deq}{}{v_{\vfail}}, I_2) \in E_\hist \implies (I_2, I_3) \in \IdlRel{v_{\qbot}}$.

\end{enumerate}
\end{definition}

The rules above are a direct lifting of the relational specification of 
$\QueueSpec$ from \secref{queue-rewrite} to frontier graphs, 
where concatenation is replaced by extensions of ideals. 
In particular, for $v_{\qbot} \in \vals_\hist \uplus \set{\qbot}$, 
$(I_1, I_2) \in \IdlRel{v_{\qbot}}$ holds iff there exists a linearization 
$\tau \in \QueueSpec$ of $I_2$ whose prefix corresponds to $I_1$, 
such that the $\enq$ operations in the prefix $\tau[:I_1]$ 
are matched by dequeue operations in the suffix $\tau[I_1:I_2]$,
except possibly for the single unmatched front value $v_{\qbot}$ 
(with $v_{\qbot}=\qbot$ indicating that no such value exists).

\begin{restatable}{lemma}{QueueReachability}
\lemlabel{queue-reachability-correctness}
  Let $\hist$ be a concurrent \queueDS history. 
  $(\fg{\hist}, I^{\init}_{\hist}, I^{\final}_{\hist})$ satisfies $\QueueSpec$-reachability
  iff there exist $I' \in \fgnodes{\hist}$ and $v_{\qbot} \in \vals_\hist \uplus \set{\qbot}$
  such that $(I', I^{\final}_{\hist}) \in \IdlRel{v_{\qbot}}$.
\end{restatable}

In \algoref{queue-lin}, we present a saturation algorithm to compute the relations
$\IdlRel{}$. As in the case of $\stackDS{}$, we maintain a table
$M : \fgnodes{\hist} \times \fgnodes{\hist} \to 2^{\vals_\hist \uplus \set{\qbot}}$
such that $M(I_1, I_2) = \setpred{c}{(I_1, I_2) \in \IdlRel{c}}$.
For succinctness, we view the clauses of \defref{queue-ideal-rel} as inducing a transition
relation $(I_1,I_2,c)\queueStep(I_3,I_4,c')$, where one rule derives
$(I_3,I_4)\in \IdlRel{c'}$ from $(I_1,I_2)\in \IdlRel{c}$.
Starting from $M(I^{\init}_{\hist}, I^{\init}_{\hist}) = \set{\qbot}$, the algorithm
performs a worklist-based saturation over $\queueStep$.
By \lemref{queue-reachability-correctness}, $\hist$ is linearizable iff there exist
$I' \in \fgnodes{\hist}$ and $c \in M(I', I^{\final}_{\hist})$.


\begin{algorithm}[t]
\caption{Linearizability monitoring for $\queueDS$}
\algolabel{queue-lin}
\myproc{$\QueueLin(\hist)$}{
\Let $\fg{\hist} \gets (\fgnodes{\hist}, \fgedges{\hist})$ be the frontier graph of $\hist$\;
\lForEach{$(I, I') \in \fgnodes{\hist} \times \fgnodes{\hist}$}{
   $M(I, I') \gets \emptyset$
}
$M(I^{\init}_{\hist}, I^{\init}_{\hist}) \gets \set{\qbot}$;
$W \gets \set{\tuple{I^{\init}_{\hist}, I^{\init}_{\hist}, \qbot}}$\;
\While{$W \neq \emptyset$}{
  remove some $\tuple{I_1,I_2,c}$ from $W$\;
  \ForEach{$\tuple{I_1,I_2,c} \queueStep \tuple{I_3,I_4,c'}$ \WHERE $c' \notin M(I_3,I_4)$}{
    $M(I_3,I_4) \gets M(I_3,I_4) \cup \set{c'}$;
    $W \gets W \cup \set{\tuple{I_3,I_4,c'}}$\;
  }
}

\Return $\exists I' \in \fgnodes{\hist} \text{ such that } M(I', I^{\final}_{\hist}) \neq \emptyset$
}
\end{algorithm}

\begin{corollary}\corlabel{queue-lin-correct}
Given a queue history $\hist$, $\QueueLin{\hist}$ returns $\true$
iff $\hist$ is linearizable.
\end{corollary}

For any fixed pair $(I_1,I_2)$, the set $M(I_1,I_2)$ has size at most $k+1$:
by the characterization above, each non-$\qbot$ value corresponds to a possible
unmatched front value at the cut $I_1$, and distinct such values arise from
different processes. Hence \algoref{queue-lin} processes at most
$O(k \cdot |\fgnodes{\hist}|^2)$ triples, and runs in
$O(k \cdot |\fgnodes{\hist}|^2)$ time.

\begin{theorem}
\thmlabel{queue-FPT}
\protect{\circledsmall{C4}}
Given a \queueDS{} history $\hist$ with $n$ operations and $k$ processes,
\algoref{queue-lin} runs in time $O(k \cdot 2^{2k} \cdot n^2)$ and returns
$\true$ iff $\hist$ is linearizable.
\end{theorem}


\section{Empirical Evaluation}
\seclabel{eval}

As shown by Gibbons et al.~\cite{Gibbons1999Lin,Gibbons2002LinJournal},
checking linearizability is \np-hard for many of the discussed data types.
Prior state-of-the-art tools for general linearizability, such as VeriLin~\cite{Jia2023},
explore a potentially exponential number of linearizations.
This can quickly render naive exploration algorithms infeasible even for modest-sized histories
that require large linearization depths~\cite{ozkan2019checking}.
The main objective of our evaluation is hence to
showcase the improved scaling performance that \fptlin{} has over VeriLin for general linearizability,
while remaining competitive with specialized tools such as LinP and LiMo for unambiguous histories~\cite{lee2025,Emmi2015}.
We generate only unambiguous histories, as required by all the above tools,
with all tools providing sound and complete linearizability results.
These are state-of-the-art tools that have been shown to outperform
other exponential-time enumeration-based algorithms~\cite{lincheck2023,lincheck2024,Lowe2017}.

\myparagraph{Implementation}
We have implemented all our algorithms in our tool \fptlin~\footnote{The source code is available at \href{https://sites.google.com/view/fptlin}{https://sites.google.com/view/fptlin}}.
\fptlin is primarily written in C++ to leverage its high performance and efficiency.
The input to \fptlin is a history specified as a set of operations,
where each operation includes the process that performs it,
along with its method and value attributes.
We use the Scal framework~\cite{Scal2015} to obtain implementations of concurrent data structures,
as well as example client code that exercises them.
Scal ships with a suite of common lock-free implementations of \queueDS and \stackDS.
For our experiments, we use Treiber stack~\cite{thomas1986systems}
and k-relaxed stack~\cite{Henzinger2013}
as linearizable and non-linearizable implementations of \stackDS, respectively.
Likewise, we use Michael-Scott queue~\cite{Michael1996}
and unbounded-size k-FIFO queue~\cite{Kirsch2013}
as linearizable and non-linearizable implementations of \queueDS, respectively.
Scal's suite of implementations, however, does not include any implementation 
of a concurrent \pqueueDS. 
To address this, we implemented a custom lock-based \pqueueDS
along with a simple client based on Scal's producer-consumer routine framework;
the histories generated by this implementation are guaranteed to be linearizable.
We generate histories of varying sizes via a producer-consumer routine,
and run all experiments on Linux 6.19.7-1 CachyOS with a 5.6\,GHz CPU and 15.46\,GB RAM.

\subsection{Scaling with History Size}

We fix the routine to use 5 producers and 5 consumers (5p5c).
We generate 100 histories, each with $N$ successful add/remove operations, where:
\begin{enumerate}
  \item $N \in [10{,}000, 1{,}000{,}000]$ at intervals of $10{,}000$ for queues and priority queues, and
  \item $N \in [50, 5{,}000]$ at intervals of $50$ for stacks.
\end{enumerate}

This corresponds to an average of $\frac{N}{10}$ values added by each producer
and removed by each consumer.
However, by allowing failed $\pop$/$\deq$ operations,
the generated histories may contain slightly more than $N$ operations.
We set a timeout of 100 seconds for the evaluation of each history,
and present the results separately for each data type implementation as follows.

\myparagraph{The \stackDS ADT \protect{\circledsmall{C5}}}
\figref{hist-n-stack} shows that \fptlin maintains a predictable
cubic scaling behavior up to around 5{,}000 operations,
with monitoring completed within 25 seconds.
On the other hand, VeriLin frequently times out for histories with more than 1{,}000 operations.
We implemented an optimized version of \algoref{stack-lin}, following the observation that each entry in the completed production table contains at most one non-terminal due to the grammar’s unambiguity.
We therefore implement each entry as an optional value, avoiding the significantly larger overheads that come with managing sets.
Additionally, this allows us to prune further exploration once a value is found for an entry.

\begin{figure}[t]
\centering

\begin{subfigure}{0.48\textwidth}
\centering
\scalebox{0.9}{
\begin{scatterplot}{operations}{time (s)}
  \addscatter{empirical/hist_n/lin_stack_fptlin.txt}{FPTLin}{color=blue,mark=square}
  \addscatter{empirical/hist_n/lin_stack_linp.txt}{LinP}{color=red,mark=o}
  \addscatter{empirical/hist_n/lin_stack_limo.txt}{LiMo}{color=orange,mark=x}
  \addscatter{empirical/hist_n/lin_stack_verilin.txt}{VeriLin}{color=green,mark=triangle}
\end{scatterplot}
}
\caption{Linearizable \stackDS (treiber)}
\end{subfigure}
\hfill
\begin{subfigure}{0.48\textwidth}
\centering
\scalebox{0.9}{
\begin{scatterplot}{operations}{time (s)}
  \addscatter{empirical/hist_n/nonlin_stack_fptlin.txt}{FPTLin}{color=blue,mark=square}
  \addscatter{empirical/hist_n/nonlin_stack_linp.txt}{LinP}{color=red,mark=o}
  \addscatter{empirical/hist_n/nonlin_stack_limo.txt}{LiMo}{color=orange,mark=x}
  \addscatter{empirical/hist_n/nonlin_stack_verilin.txt}{VeriLin}{color=green,mark=triangle}
\end{scatterplot}
}
\caption{Non-linearizable \stackDS (kstack)}
\end{subfigure}

\caption{\fptlin{}'s scalability w.r.t history size (\stackDS)}
\figlabel{hist-n-stack}
\end{figure}

\begin{figure}[t]
\centering

\begin{subfigure}{0.48\textwidth}
\centering
\scalebox{0.9}{
\begin{scatterplot}{operations}{time (s)}
  \addscatter{empirical/hist_n/lin_queue_fptlin.txt}{FPTLin}{color=blue,mark=square}
  \addscatter{empirical/hist_n/lin_queue_linp.txt}{LinP}{color=red,mark=o}
  \addscatter{empirical/hist_n/lin_queue_limo.txt}{LiMo}{color=orange,mark=x}
  \addscatter{empirical/hist_n/lin_queue_verilin.txt}{VeriLin}{color=green,mark=triangle}
\end{scatterplot}
}
\caption{Linearizable \queueDS (ms)}
\end{subfigure}
\hfill
\begin{subfigure}{0.48\textwidth}
\centering
\scalebox{0.9}{
\begin{scatterplot}{operations}{time (s)}
  \addscatter{empirical/hist_n/nonlin_queue_fptlin.txt}{FPTLin}{color=blue,mark=square}
  \addscatter{empirical/hist_n/nonlin_queue_linp.txt}{LinP}{color=red,mark=o}
  \addscatter{empirical/hist_n/nonlin_queue_limo.txt}{LiMo}{color=orange,mark=x}
  \addscatter{empirical/hist_n/nonlin_queue_verilin.txt}{VeriLin}{color=green,mark=triangle}
\end{scatterplot}
}
\caption{Non-linearizable \queueDS (us-kfifo)}
\end{subfigure}

\caption{\fptlin{}'s scalability w.r.t history size (\queueDS)}
\figlabel{hist-n-queue}
\end{figure}

\begin{wrapfigure}{r}{0.45\textwidth}
\scalebox{0.9}{
\begin{scatterplot}{operations}{time (s)}
  \addscatter{empirical/hist_n/lin_pqueue_fptlin.txt}{FPTLin}{color=blue,mark=square}
  \addscatter{empirical/hist_n/lin_pqueue_linp.txt}{LinP}{color=red,mark=o}
\end{scatterplot}
}
\vspace{-0.2in}
\caption{\fptlin{}'s scalability w.r.t history size (\pqueueDS{})}
\figlabel{hist-n-pqueue}
\end{wrapfigure}

\myparagraph{The \queueDS ADT \protect{\circledsmall{C5}}}
\figref{hist-n-queue} shows that our \queueDS algorithm scales linearly in practice.
This attributes to small gaps between $k$-th invocation and $k$-th response of $\deq$ operations in the histories.
This assumption is also practical to make for most lock-free implementations of the queue data type.
\figref{hist-n-queue} shows that \fptlin{} comfortably handles histories with slightly over 1,000,000 operations, completing monitoring in approximately 4 seconds.
The implemented algorithm is loosely based on \algoref{queue-lin}, involving additional constant factor optimizations.
Similar to the implementation of the algorithm for stack, we implement each entry of the dynamic programming table as an optional value instead of sets to avoid large overheads.
We also observe that given a queue history $\hist$, maintaining two separate frontier graphs for $\proj{\hist}{\enq}$ and $\proj{\hist}{\peek, \deq}$ respectively can help greatly reduce the size of the dynamic programming table.

\myparagraph{The \pqueueDS ADT \protect{\circledsmall{C5}}}
We 
implemented
a framework for order oblivious data types, and provided specializations for set, read-modify-write register, semaphore, and priority queue data types.
We benchmark our algorithm for priority queue since it exhibits the worst time complexity among the implemented specializations.
We implemented a lock-based priority queue to generate histories of guaranteed linearizability.
\figref{hist-n-pqueue} shows that our framework is able to scale up to 1,000,000 operations within 2 seconds,
staying competitive with LinP's runtime.

\subsection{Scalability with Thread Count}

\begin{wrapfigure}{r}{0.5\textwidth}
\vspace{-0.2in}
\scalebox{0.9}{
\begin{scatterplot}{thread count}{time taken per operation ($\mu$s)}
  \addscatter{empirical/hist_p/lin_pqueue_fptlin.txt}{FPTLin}{color=blue,mark=square}
  \addscatter{empirical/hist_p/lin_pqueue_linp.txt}{LinP}{color=red,mark=o}
\end{scatterplot}
}
\vspace{-0.1in}
\caption{\fptlin{}'s scalability w.r.t thread count (\pqueueDS)}
\figlabel{hist-p-pqueue}
\vspace{-0.1in}
\end{wrapfigure}
For all data points,
 we fix the history size to $\approx 1 000 000$ operations for queue and priority queue,
and $\approx 5000$ operations for stack.
\fptlin{} supports a maximum of 64 threads in a history.
Hence, we have 5 runs for 32 data points with configurations ranging from
1p1c to 32p32c with 1p1c increments.

\begin{figure}[t]
\centering

\begin{subfigure}{0.48\textwidth}
\centering
\scalebox{0.9}{
\begin{scatterplot}{thread count}{time taken per operation ($\mu$s)}
  \addscatter{empirical/hist_p/lin_stack_fptlin.txt}{FPTLin}{color=blue,mark=square}
  \addscatter{empirical/hist_p/lin_stack_linp.txt}{LinP}{color=red,mark=o}
  \addscatter{empirical/hist_p/lin_stack_limo.txt}{LiMo}{color=orange,mark=x}
  \addscatter{empirical/hist_p/lin_stack_verilin.txt}{VeriLin}{color=green,mark=triangle}
\end{scatterplot}
}
\caption{Linearizable \stackDS (treiber)}
\end{subfigure}
\hfill
\begin{subfigure}{0.48\textwidth}
\centering
\scalebox{0.9}{
\begin{scatterplot}{thread count}{time taken per operation ($\mu$s)}
  \addscatter{empirical/hist_p/nonlin_stack_fptlin.txt}{FPTLin}{color=blue,mark=square}
  \addscatter{empirical/hist_p/nonlin_stack_linp.txt}{LinP}{color=red,mark=o}
  \addscatter{empirical/hist_p/nonlin_stack_limo.txt}{LiMo}{color=orange,mark=x}
  \addscatter{empirical/hist_p/nonlin_stack_verilin.txt}{VeriLin}{color=green,mark=triangle}
\end{scatterplot}
}
\caption{Non-linearizable \stackDS (kstack)}
\end{subfigure}

\caption{\fptlin{}'s scalability w.r.t thread count (\stackDS)}
\figlabel{hist-p-stack}
\end{figure}

\begin{figure}[t]
\centering

\begin{subfigure}{0.48\textwidth}
\centering
\scalebox{0.9}{
\begin{scatterplot}{thread count}{time taken per operation ($\mu$s)}
  \addscatter{empirical/hist_p/lin_queue_fptlin.txt}{FPTLin}{color=blue,mark=square}
  \addscatter{empirical/hist_p/lin_queue_linp.txt}{LinP}{color=red,mark=o}
  \addscatter{empirical/hist_p/lin_queue_limo.txt}{LiMo}{color=orange,mark=x}
  \addscatter{empirical/hist_p/lin_queue_verilin.txt}{VeriLin}{color=green,mark=triangle}
\end{scatterplot}
}
\caption{Linearizable \queueDS (ms)}
\end{subfigure}
\hfill
\begin{subfigure}{0.48\textwidth}
\centering
\scalebox{0.9}{
\begin{scatterplot}{thread count}{time taken per operation ($\mu$s)}
  \addscatter{empirical/hist_p/nonlin_queue_fptlin.txt}{FPTLin}{color=blue,mark=square}
  \addscatter{empirical/hist_p/nonlin_queue_linp.txt}{LinP}{color=red,mark=o}
  \addscatter{empirical/hist_p/nonlin_queue_limo.txt}{LiMo}{color=orange,mark=x}
  \addscatter{empirical/hist_p/nonlin_queue_verilin.txt}{VeriLin}{color=green,mark=triangle}
\end{scatterplot}
}
\caption{Non-linearizable \queueDS (us-kfifo)}
\end{subfigure}

\caption{\fptlin{}'s scalability w.r.t thread count (\queueDS)}
\figlabel{hist-p-queue}
\end{figure}

\myparagraph{The \stackDS ADT \protect{\circledsmall{C5}}}
\figref{hist-p-stack} shows that \fptlin{} generally performs well with lower thread count,
and maintains a stable runtime as thread count grows for a linearizable implementation of \stackDS,
while VeriLin again often timed out given any configuration.
Although \fptlin{} began timing out for higher thread count for the non-linearizable stack implementation,
which is acceptable for the purpose of verification
as we only require only one result for proof of non-linearizability.

\myparagraph{The \queueDS ADT \protect{\circledsmall{C5}}} 
\figref{hist-p-queue} shows \fptlin{}'s \queueDS algorithm having lower constant factor than that of
the \stackDS algorithm, and better runtime predictability than that of VeriLin.

\myparagraph{The \pqueueDS ADT \protect{\circledsmall{C5}}}
\figref{hist-p-pqueue} shows the shortcoming of \fptlin{} as it is outperformed by LinP
in runtime per operation for \pqueueDS.
\fptlin{} exhibits obvious exponential growth with respect of thread count at the start of the graph.
In practice, being able to process an operation around 100 $\mu$s allows
\fptlin{} to remain a viable option for general linearizability monitoring for OODTs. 


\section{Related Work and Discussion}
\seclabel{related}

\myparagraph{Linearizability and verifying implementations}
Herlihy and Wing introduced linearizability in their seminal work~\cite{Wing1990}.
Since then, a variety of techniques for \emph{verifying} linearizability of
implementations have been developed, e.g., based on simulation, refinement, and
abstraction~\cite{Abdulla2013,Amit2007,Bouajjani2015a}. Bouajjani et
al.~\cite{Bouajjani2013} showed that verification is undecidable in general,
with decidable fragments arising under bounded concurrency~\cite{Alur2000,Hamza2015}
or restrictions on invocation patterns~\cite{vcerny2010model}.
Dongol and Derrick~\cite{DongolDerrickSurvey2015} survey these approaches and
their proof principles, while more recent work develops proof systems designed
for automated or mechanised reasoning about linearizability~\cite{LinearizabilityForward2024,AspectOriented2013}.

\myparagraph{Linearizability monitoring}
Linearizability \emph{monitoring} dates back to early exponential-time
algorithms~\cite{Wing1993}. Gibbons and Korach~\cite{Gibbons1997} established
$\np$-hardness for \registerDS and showed that, under natural restrictions on
histories, monitoring becomes tractable. 
Follow-up work extended these hardness
and tractability results to \stackDS, \queueDS, \pqueueDS, and other 
types~\cite{Gibbons1999Lin,Gibbons2002LinJournal,Emmi2015,bouajjani2018reducing,Enea2017},
often under unambiguity or differentiation assumptions. More recently, log-linear
algorithms have been obtained for \stackDS, \queueDS, \pqueueDS under ambiguity and
for (multi)sets for arbitrary histories~\cite{lee2025,Abdulla2025}. A complementary line of work
builds general-purpose monitors and testing frameworks that explore (potentially
exponential) sets of executions or histories to detect violations~\cite{lincheck2023,lincheck2024,Jia2023,Burckhardt2010,Lowe2017,Zhang2015,Horn2015},
sometimes combined with symbolic reasoning~\cite{relinche2025,Emmi2015} or
bounded-depth heuristics without full correctness guarantees~\cite{ozkan2019checking}.

The work of Lowe~\cite{Lowe2017} comes particularly close to ours.
Lowe proposes a generic graph-search algorithm for linearizability monitoring,
improving on the tree-based search of Wing and Gong~\cite{Wing1993}
by caching previously visited configurations.
Each configuration in Lowe's algorithm consists of a set of operations
that have been linearized so far together with a state of the sequential
specification object.
Since operations are added only when minimal, the set of linearized operations
forms an ideal of the history in our sense.
Thus, Lowe's algorithm can be viewed as exploring a graph whose nodes
consist of pairs of 
(i) a frontier of the history and (ii) a reachable state of the specification object.
Lowe's complexity analysis is expressed in terms of a parameter $B_{p,n}$,
which bounds the number of specification states reachable for a fixed choice
of linearized operations, when the history has $n$ operations and $p$ processes.
In this light, \lemref{frontier-graph-size} provides a complementary,
purely structural bound on the number of possible frontiers of a history.
Combining the two viewpoints, it follows that the size of the search space
explored by Lowe's graph-search algorithm is bounded by
the number of frontiers (as given by \lemref{frontier-graph-size})
multiplied by $B_{p,n}$.
Thus, our lemma can be seen as making explicit the history-dependent component
of Lowe's complexity bound.
This decomposition highlights a clean separation between the combinatorial complexity of the history and the semantic complexity of the ADT.
In particular, while \cite{Lowe2017} derives bounds on $B_{p,n}$ for specific data structures
such as \registerDS{} and \mapDS{} via ad hoc state-counting arguments,
it does not identify a general semantic class of ADTs for which such bounds hold.
In contrast, in \secref{oodt-lin}, we define the class of $\alpha$-OODTs and prove
a uniform FPT bound (\thmref{alpha-OODT-FPT}) for all ADTs in this class.
Another key difference is that $B_{p,n}$ depends on the concrete
representation of the sequential specification object, and may be large
even when the underlying ADT has a simple semantic structure
(for example, due to auxiliary state such as timestamps).
In contrast, our notion of $\alpha$-OODT is defined semantically,
using the indistinguishability relation $\sim_{\ADT}$,
and yields a representation-independent bound.
In particular, $\alpha$-OODTs are not exactly those ADTs for which
$B_{p,n} \le \alpha$, since $B_{p,n}$ counts concrete specification states,
while $\alpha$ bounds the number of semantic equivalence classes.
Finally, our results for stacks and queues in fact, do not follow from this
generic algorithm of~\cite{Lowe2017}.
In fact, for \queueDS{}, the analysis of \cite{Lowe2017}, explicitly observes 
factorial growth in the number of states [``at least $p!^{m/p}$''], 
leading to exponential behavior; the same argument applies to \stackDS{}. 
Our frontier-graph and language-reachability framework instead yields FPT bounds for \stackDS (via context-free reachability) and polynomial-time monitoring for \queueDS on frontier graphs.

\myparagraph{Other works}
Our formulation of monitoring as language reachability on frontier graphs connects
directly to a broader body of work on graph reachability with language
con,straints. Melski and Reps~\cite{Melski1997} gave a dynamic-programming
algorithm for all-pairs context-free reachability, while Azimov et
al.~\cite{Azimov2018} reduced context-free reachability to matrix operations.
Fine-grained complexity results for CFL reachability~\cite{koutris2023fine}
show conditional lower bounds even for restricted grammars, and recent work
extends these ideas to multiple context-free languages~\cite{Conrado2025},
which can likewise be used as sequential specifications in our framework.
More broadly, our work complements works on runtime predictive analysis~\cite{Mathur2020,Kini2017WCP,mathur2021optimal,Mathur2018SHB,OSR2024Shi,TuncDeadlock2023} and consistency checking~\cite{tuncc2023optimal,chakraborty2024hard,Gibbons1997,Shi2025MessagePassing}
for which FPT algorithms have largely remained elusive, except for some recent results~\cite{farzan2026parametrizingreadsfromequivalencepredictive,Ang2025GPrefix}.


\section{Conclusions and Future Work}
\seclabel{conclusions}

We presented a reduction from linearizability monitoring to a graph reachability problem
and developed a fixed-parameter tractable (FPT) algorithm that exploits the tractability of reachability on frontier graphs.
We then investigated three classes of languages (corresponding to data types) for this reachability problem,
each yielding polynomial-time algorithms:
(i) $\alpha$-Order Oblivious Data Types, including \registerDS{}s, counting semaphores, and \pqueueDS{}s;
(ii) Context-Free Data Types, that includes the \stackDS ADT; and
(iii) the \queueDS ADT (treated as a distinct class).
To our knowledge, these constitute the first FPT linearizability 
monitoring algorithms for \stackDS, \queueDS, \pqueueDS,
and a broad family of additional data types.
We implemented our approach in our tool \fptlin
and demonstrated that it scales to million-operation histories
while outperforming existing state-of-the-art general linearizability monitoring tools.
Beyond these results, our work opens up several directions for future research.
An interesting avenue is to generalize our approach to \queueDS to richer concurrent objects such as deques.
Another direction includes algorithmic characterization of when ADTs belong
to the proposed $\alpha$-OODT classes.
Finally, we would like to explore if heuristics such as saturation~\cite{tuncc2023optimal,Shi2025MessagePassing}
improve the performance of the proposed algorithms. 

\section*{Data Availability Statement}
Our tool \fptlin is publicly available at~\cite{toolLink}, and the version used in this work is available at~\cite{artifactLink}.

\begin{acks}
This work is partially supported by the National Research Foundation, Singapore, and Cyber Security Agency of Singapore under its National Cybersecurity R\&D Programme (Fuzz Testing <NRF-NCR25-Fuzz-0001>). Any opinions, findings and conclusions, or recommendations expressed in this material are those of the author(s) and do not reflect the views of National Research Foundation, Singapore, and Cyber Security Agency of Singapore.
\end{acks}

\bibliographystyle{ACM-Reference-Format}
\bibliography{references}


\appendix

\pagebreak

\section{Sub-cubic time stack reachability}
\applabel{stack-subcubic}

\myparagraph{One-step Matrix}
Given a stack history $\hist$,
we define an initial $|\fgnodes{\hist}| \times |\fgnodes{\hist}|$ \emph{one-step matrix} of $\hist$, $M$,
indexed by ideals of $\hist$, where each entry is a set of non-terminals,
by initializing entries $M(I_1, I_2)$ where $(I_1, \abs{o}, I_2) \in\fgedges{\hist}$ with corresponding non-terminals,
and empty set otherwise.

\myparagraph{Reducing Linearizability to Transitive Closure}
We also define a non-commutative multiplication $\otimes$ between sets of nonterminals $X_1$ and $X_2$ where $X_1 \otimes X_2 = \setpred{A}{A \rightarrow_\stackDS B_1B_2, B_1 \in X_1, B_2 \in X_2}$. In a similar fashion, we define the same multiplication for matrix to $M''' = M' \otimes M''$ to be analagous to numerical matrix multiplication but with numerical multiplication substituted by $\otimes$ between sets of nonterminals, and accumulation by union. That is,
\[
M'''(I_1, I_2) = \bigcup_{I_3 \in \idealsOf{\hist}}{M'(I_1, I_3) \otimes M''(I_3, I_2)}
\]
The transitive closure of $M$ can be defined as:
\[
M^+ = M^{(1)} \cup M^{(2)} \cup \dots
\]
where
\[
M^{(i)} = \bigcup^{i-1}_{j=1}{M^{(j)} \otimes M^{(i-j)}}
\]
Notice that $M^+$ is exactly the resultant recognition table of \algoref{stack-lin}.
The construction of $M$ is not more expensive than $O(|\fgnodes{\hist}|^2)$.
We denote here $\streach{N, k}$ and $\transclo{N, k}$ as the time complexity of \stackDS reachability and transitive closure respectively
for a given frontier graph of $N$ vertices and $k$ processes.
Since the reduction involves the construction of $M$, we have the following:

\begin{lemma}\lemlabel{l-t-reduce}
  $\streach{N, k} = \transclo{N, k} + O(N^2)$.
\end{lemma}

\myparagraph{Reducing Transitive Closure to Multiplication}
Reader may verify that arranging index in increasing ideal size naturally enforces $M$ to be an upper triangular matrix.
For which, we can trivially apply Valiant's results on reducing transitive closure to matrix multiplication defined above.
We denote $\matmul{N, k}$ as the time complexity of matrix multiplication given matrix size $N$.

\begin{lemma}\lemlabel{t-mm-reduce}
  $\transclo{N, k} = \matmul{N, k} \cdot O(\log{N})$.
\end{lemma}
\begin{proof}
  See~\cite{Valiant1974}.
\end{proof}

\myparagraph{Reducing Multiplication to Boolean Multiplication}
As per Valiant's reduction, a matrix multiplication for a given context-free grammar $G$ is simulated using no more than $|G|^2$ boolean matrix multiplications.
As for our case of monitoring linearizability of stack histories,
the size of the grammar grows with the size of the input history, which is undesirable.
Fortunately, there are lots of unnecessary work done from Valiant's reduction that can be optimized away.
Intuitively, we require no more than $k + 1$ non-terminals per ideal to account for a first term in the matrix multiplication,
and no more than $k$ non-terminals per ideal for the second term.
Hence, we can simulate a single matrix multiplication of size $N$ with the multiplication of a boolean matrix of size $2kN + 2N$.

\begin{restatable}{lemma}{MMReduceBMM}\lemlabel{mm-bmm-reduce}
    $\matmul{N, k} = \boolmatmul{2kN + 2N} + O(kN^2)$.
\end{restatable}
\begin{proof}
  See \appref{stack-proof}.
\end{proof}

Finally, with the combination of \lemref{l-t-reduce}, \lemref{t-mm-reduce} and \lemref{mm-bmm-reduce},
we arrive at an algorithm asymptotically faster than \algoref{stack-lin}~\cite{Alman2025},
assuming $k$ to be constant.

\begin{lemma}
  $\streach{N, k} = O(k^\omega N^\omega)$, where $\omega$ is the boolean matrix multiplication exponent.
\end{lemma}
\begin{proof}
  Assume $k \geq 1$, we have:
  \begin{align*}
    \streach{N, k} &= \transclo{N, k} + O(N^2)\\
    &= \matmul{N, k} \cdot O(\log{N}) + O(N^2)\\
    &= \boolmatmul{2kN + 2N} \cdot O(\log{N}) + O(kN^2) \cdot O(\log{N}) + O(N^2)\\
    &= O((4kN)^\omega) \cdot O(\log{N}) + O(kN^2) \cdot O(\log{N}) + O(N^2)\\
    &= O((kN)^\omega)
  \end{align*}
\end{proof}

\section{Proofs for Stack Section}
\applabel{stack-proof}

\begin{lemma}\lemlabel{tran-clo-card}
  Given a non-empty well-matched stack history $\hist$ with no failed operations.
  Let $M^+$ be the transitive closure of the one-step matrix $M$ of $\hist$.
  For any ideal $I_1 \in \idealsOf{\hist}$:
  \begin{enumerate}
    \item $|\bigcup_{I_2 \in \idealsOf{\hist}}\set{Push(v),Peek(v),T(\varepsilon) \in M^+(I_1, I_2)}| \leq |\procs_\hist| + 1$, and
    \item $|\bigcup_{I_2 \in \idealsOf{\hist}}\set{T(\varepsilon),T(v) \in M^+(I_2, I_1)}| \leq |\procs_\hist| + 1$.
  \end{enumerate}
\end{lemma}
\begin{proof}
  \begin{enumerate*}
    \item
    Observe that $Push(v)$ and $Peek(v)$ generates to a single abstract operation $\abs{o}$ such that
    $(I_1, \abs{o}, I_2) \in \fgedges{\hist}$.
    There are only be a maximum of $|\procs_\hist|$ such abstract operation.
    Including $T(\varepsilon)$ yields a maximum of $|\procs_\hist|+1$ symbols.
    \item
    $T(v)$ only generates sequences ending with $\opr{\pop}{v}{\vok}$.
    Again, there can be only a maximum of $|\procs_\hist|$ such operations such that
    $(I_3, \opr{\pop}{v}{\vok}, I_1) \in \fgedges{\hist}$ for some ideal $I_3$ and value $v$.
    The conclusion follows.
  \end{enumerate*}
\end{proof}

\MMReduceBMM*
\begin{proof}
  Suppose we want to compute the product of matrices $M_1 \otimes M_2 = M_3$.
  First, we construct a representative boolean matrix $M'_1$ for a given $M_1$.
  Notice that in the production rules, the right non-terminal must be of the form $T(\varepsilon)$ or $T(v)$.
  Consider the entry, $M_1(I_1, I_2)$ where $T(\varepsilon), T(v) \in M_1(I_1, I_2)$.
  We can see that there are at most $k + 1$ such values for any entries ending with $I_2$ by \lemref{tran-clo-card},
  we associate these values with $I_2$.
  The membership of $T(v)$ in $M_1(I_1, I_2)$ is now presented by a new entry $M'_1(I_1, \tuple{I_2, T(v)})$ being set.
  Similarly, notice that the left non-terminal must be of the form $Push(v)$, $Peek(v)$ or $T(\varepsilon)$.
  Again, consider the entry, $M_1(I_3, I_4)$.
  There are at most $k + 1$ such values for any entries starting with $I_3$ by \lemref{tran-clo-card}.
  Hence, the membership of any non-terminal $T \in M_1(I_3, I_4)$ is now represented by a new entry $M'_1(\tuple{I_3, T}, I_4)$ being set.
  By enforcing the presence of all possible indexes for both dimensions,
  we get a boolean matrix of size $(N + N(2k+1)) = (2kN + 2N)$.
  After constructing $M'_2$ from $M_2$ the same way, we compute $M'_1 \otimes M'_2 = M'_3$ using the black-box BMM procedure.
  We can now see that $M_3$ can be efficiently constructed from $M'_3$ by iterating through each entry of $M_3$,
  and set $M_3(I_1, I_2) = \bigcup\setpred{P}{P \rightarrow_\stackDS LR, M'_3(\tuple{I_1, L}, \tuple{I_2, R}) = 1}$.
  This conversion back to $M_3$ be done in $O(kN^2)$.
\end{proof}

\section{Proofs from \secref{queue-adt}}
\applabel{queue-proofs}


\subsection{Proofs from \secref{queue-rewrite}}
\applabel{queue-rewrite}

Here, we prove \lemref{queue-rewrite-correctness}. We first present
some auxiliary definitions and lemmas.
In the following, we will say that an operation $o$ is a \emph{successful}
operation if $\abs{o} \notin \setpred{\opr{m}{}{\vfail}}{m \in \set{\peek, \deq}}$.

\newcommand{\Survives}[1]{\mathsf{Survives}_{#1}}

\myparagraph{Surviving enqueue operations}
Fix a sequential queue history $\tau \in \QueueSpec{}$.
For every successful dequeue operation $o_{\deq}$ in $\tau$,
let $\match_{\tau}(o_{\deq})$ denote the unique earlier enqueue
operation matched to $o_{\deq}$.
Let $\tau = \alpha \cdot \beta$.
An enqueue operation $o_{\enq} \in \alpha$ is \emph{surviving}
if it is not matched by any successful dequeue in $\tau$.
We define the set of surviving enqueue operations from $\alpha$ as
\[
\mathsf{Survivors}(\alpha,\beta)
=
\setpred{o_{\enq} \in \alpha}{
\text{$o_{\enq}$ is unmatched in } \alpha \cdot \beta
}.
\]

Using this notation, for $v \in \vals$, we define:
\begin{align*}
\Survives{\qbot}(\alpha,\beta)
\quad\text{iff}\quad
& \alpha \cdot \beta \in \QueueSpec{},
\\
& \mathsf{Survivors}(\alpha,\beta)=\emptyset,
\\
& \text{every successful } \peek \text{ or } \deq \text{ operation in } \beta
\text{ observes/matches an enqueue in } \alpha,
\\
& \text{and every failed } \peek \text{ or } \deq \text{ operation in } \beta
\text{ occurs at the beginning of } \beta,
\\[0.5em]
\Survives{v}(\alpha,\beta)
\quad\text{iff}\quad
& \alpha \cdot \beta \in \QueueSpec{},
\\
& \mathsf{Survivors}(\alpha,\beta)=\set{o_{\enq}},
\\
& \valOf{o_{\enq}} = v \text{ and } o_{\enq} \text{ is the last enqueue in } \alpha,
\\
& \text{every successful } \peek \text{ or } \deq \text{ operation in } \beta
\text{ observes/matches an enqueue in } \alpha,
\\
& \text{and every failed } \peek \text{ or } \deq \text{ operation in } \beta
\text{ occurs at the beginning of } \beta.
\end{align*}

\begin{lemma}
\lemlabel{queue-rewrite-invariant}
For every $v_{\qbot} \in \vals \cup \set{\qbot}$ and $\alpha,\beta \in \Sigma^*$,
\[
(\alpha,\beta) \in R_{v_{\qbot}}
\quad\Longleftrightarrow\quad
\Survives{v_{\qbot}}(\alpha,\beta).
\]
\end{lemma}

\begin{proof}
We prove the two directions separately.

\medskip

\smallskip
\noindent
$(\Rightarrow)$ We prove by induction on the derivation of
$(\alpha,\beta)\in R_{v_{\qbot}}$.

\begin{description}
\item[Base case.]
For the rule $(\epsilon,\epsilon)\in R_{\qbot}$, we have
$\epsilon\in\QueueSpec$, there are no enqueues, and there are no operations
in the suffix. Hence $\Survives{\qbot}(\epsilon,\epsilon)$ holds.

\item[Inductive step.]
We consider the last rule used in the derivation.

\begin{enumerate}
\item Suppose the last rule is
\[
(\alpha,\epsilon)\in R_{\qbot}
\Rightarrow
(\alpha,\opr{m}{}{\vfail})\in R_{\qbot}
\quad(m\in\set{\peek,\deq}).
\]
By induction, $\Survives{\qbot}(\alpha,\epsilon)$ holds. Thus $\alpha$ is
legal and all enqueues in $\alpha$ are matched, so appending the failed
operation $\opr{m}{}{\vfail}$ is legal. The failed operation matches no
enqueue and occurs at the beginning of the suffix. Hence
$\Survives{\qbot}(\alpha,\opr{m}{}{\vfail})$.

\item Suppose the last rule is
\[
(\alpha,\beta)\in R_v
\Rightarrow
(\alpha,\beta\cdot\opr{\peek}{}{v})\in R_v.
\]
By induction, $\Survives{v}(\alpha,\beta)$ holds. Appending
$\opr{\peek}{}{v}$ is legal, since the unique surviving enqueue from
$\alpha$ has value $v$. A peek does not change matching, and the new
successful peek observes an enqueue from $\alpha$. Hence
$\Survives{v}(\alpha,\beta\cdot\opr{\peek}{}{v})$.

\item Suppose the last rule is
\[
(\alpha,\beta)\in R_v
\Rightarrow
(\alpha,\beta\cdot\opr{\deq}{}{v})\in R_{\qbot}.
\]
By induction, $\Survives{v}(\alpha,\beta)$ holds. Appending
$\opr{\deq}{}{v}$ matches the unique surviving enqueue from $\alpha$,
so all enqueues in $\alpha$ become matched. The new successful dequeue
matches an enqueue from $\alpha$, and the failed-operation condition is
unchanged. Hence $\Survives{\qbot}(\alpha,\beta\cdot\opr{\deq}{}{v})$.

\item Suppose the last rule is
\[
(\alpha,\beta)\in R_{v_{\qbot}}
\Rightarrow
(\alpha,\beta\cdot\opr{\enq}{v}{\vok})\in R_{v_{\qbot}}.
\]
By induction, $\Survives{v_{\qbot}}(\alpha,\beta)$ holds. Appending an
enqueue is always legal, does not change which enqueues in $\alpha$ are
matched, and introduces no successful or failed operation in the suffix.
Hence $\Survives{v_{\qbot}}(\alpha,\beta\cdot\opr{\enq}{v}{\vok})$.

\item Suppose the last rule is
\[
(\alpha,\opr{\enq}{v}{\vok}\cdot\beta)\in R_{\qbot}
\Rightarrow
(\alpha\cdot\opr{\enq}{v}{\vok},\beta)\in R_v.
\]
By induction, $\Survives{\qbot}(\alpha,\opr{\enq}{v}{\vok}\cdot\beta)$
holds. Hence the full sequence
$\alpha\cdot\opr{\enq}{v}{\vok}\cdot\beta$ is legal, all enqueues in
$\alpha$ are matched, and every successful peek/dequeue in
$\opr{\enq}{v}{\vok}\cdot\beta$ observes/matches an enqueue from $\alpha$.
After moving $\opr{\enq}{v}{\vok}$ to the prefix, it becomes the unique
surviving enqueue from the prefix; it has value $v$ and is the last enqueue
in the new prefix. The remaining suffix $\beta$ still satisfies the
successful-operation and failed-operation conditions. Thus
$\Survives{v}(\alpha\cdot\opr{\enq}{v}{\vok},\beta)$.

\item Suppose the last rule is
\[
(\alpha,\opr{\peek}{}{v_{\vfail}}\cdot\beta)\in R_{v_{\qbot}}
\Rightarrow
(\alpha\cdot\opr{\peek}{}{v_{\vfail}},\beta)\in R_{v_{\qbot}}.
\]
By induction,
$\Survives{v_{\qbot}}(\alpha,\opr{\peek}{}{v_{\vfail}}\cdot\beta)$ holds.
Moving a peek from the beginning of the suffix to the end of the prefix
does not change the full sequence or any matching of enqueues. Moreover,
successful operations remaining in $\beta$ still observe enqueues from the
new prefix, and the failed-operation condition is preserved. Hence
$\Survives{v_{\qbot}}(\alpha\cdot\opr{\peek}{}{v_{\vfail}},\beta)$.

\item Suppose the last rule is
\[
(\alpha,\opr{\deq}{}{v_{\vfail}}\cdot\beta)\in R_{v_{\qbot}}
\Rightarrow
(\alpha\cdot\opr{\deq}{}{v_{\vfail}},\beta)\in R_{v_{\qbot}}.
\]
By induction,
$\Survives{v_{\qbot}}(\alpha,\opr{\deq}{}{v_{\vfail}}\cdot\beta)$ holds.
Moving a dequeue from the beginning of the suffix to the end of the prefix
does not change the full sequence. If it is successful, it already
matched an enqueue from $\alpha$, hence from the enlarged prefix; if it is
failed, the failed-operation condition is unaffected for the remaining
suffix. Thus $\Survives{v_{\qbot}}(\alpha\cdot\opr{\deq}{}{v_{\vfail}},\beta)$.
\end{enumerate}
\end{description}

This completes the proof of the forward direction.

\smallskip
\noindent
$(\Leftarrow)$ We prove the converse by induction on the lexicographic
measure $(|\alpha|+|\beta|,|\alpha|)$.

\begin{description}
\item[Base case.]
If $\alpha=\beta=\epsilon$, then $\Survives{v_{\qbot}}(\epsilon,\epsilon)$
implies $v_{\qbot}=\qbot$, and the base rule gives
$(\epsilon,\epsilon)\in R_{\qbot}$.

\item[Inductive step.]
Assume $\Survives{v_{\qbot}}(\alpha,\beta)$.

\begin{enumerate}
\item Suppose $\beta=\epsilon$ and $\alpha=\alpha'\cdot o$.
If $o=\opr{\enq}{v}{\vok}$, then $v_{\qbot}=v$ and
$\Survives{\qbot}(\alpha',\opr{\enq}{v}{\vok})$ holds. By induction,
$(\alpha',\opr{\enq}{v}{\vok})\in R_{\qbot}$, and the register-front rule gives
$(\alpha'\cdot\opr{\enq}{v}{\vok},\epsilon)\in R_v$.
If $o$ is a $\peek$ or $\deq$ operation, then
$\Survives{v_{\qbot}}(\alpha',o)$ holds. By induction,
$(\alpha',o)\in R_{v_{\qbot}}$, and the corresponding skip rule gives
$(\alpha'\cdot o,\epsilon)\in R_{v_{\qbot}}$.

\item Suppose $\beta=\beta'\cdot \opr{\enq}{v}{\vok}$.
Then $\Survives{v_{\qbot}}(\alpha,\beta')$ holds. By induction,
$(\alpha,\beta')\in R_{v_{\qbot}}$, and the deferred-enqueue rule gives
$(\alpha,\beta'\cdot\opr{\enq}{v}{\vok})\in R_{v_{\qbot}}$.

\item Suppose $\beta=\beta'\cdot\opr{m}{}{\vfail}$, where
$m\in\set{\peek,\deq}$. Since failed operations in $\beta$ occur only at the
beginning of $\beta$, we have $\beta'=\epsilon$. Also $v_{\qbot}=\qbot$ and
$\Survives{\qbot}(\alpha,\epsilon)$ holds. By induction,
$(\alpha,\epsilon)\in R_{\qbot}$, and the failed-operation rule gives
$(\alpha,\opr{m}{}{\vfail})\in R_{\qbot}$.

\item Suppose $\beta=\beta'\cdot\opr{\peek}{}{v}$ for $v\in\vals$.
The final peek observes an enqueue in $\alpha$, so $v_{\qbot}=v$ and
$\Survives{v}(\alpha,\beta')$ holds. By induction,
$(\alpha,\beta')\in R_v$, and the peek-front rule gives
$(\alpha,\beta'\cdot\opr{\peek}{}{v})\in R_v$.

\item Suppose $\beta=\beta'\cdot\opr{\deq}{}{v}$ for $v\in\vals$.
The final dequeue matches an enqueue in $\alpha$, so
$\Survives{v}(\alpha,\beta')$ holds. By induction,
$(\alpha,\beta')\in R_v$, and the dequeue-front rule gives
$(\alpha,\beta'\cdot\opr{\deq}{}{v})\in R_{\qbot}$.
\end{enumerate}
\end{description}
\end{proof}

We now move to the proof of \lemref{queue-rewrite-correctness}:

\QueueRewriteSystem*

\begin{proof}
We use \lemref{queue-rewrite-invariant}.

For $\supseteq$, suppose $\tau=\alpha\cdot\beta$ and
$(\alpha,\beta)\in R_{v_{\qbot}}$ for some $v_{\qbot}$.
By \lemref{queue-rewrite-invariant},
$\Survives{v_{\qbot}}(\alpha,\beta)$ holds, whose definition implies
$\alpha\cdot\beta\in\QueueSpec{}$. Hence $\tau\in\QueueSpec{}$.

For $\subseteq$, let $\tau\in\QueueSpec{}$.
If every enqueue in $\tau$ is matched by a successful dequeue, then
$\Survives{\qbot}(\tau,\epsilon)$ holds. Hence, by
\lemref{queue-rewrite-invariant}, $(\tau,\epsilon)\in R_{\qbot}$.

Otherwise, let $o$ be the first enqueue operation in $\tau$ that is not
matched by any successful dequeue in $\tau$. Write
$\tau=\alpha\cdot\beta$, where $\alpha$ is the prefix ending at $o$.
Let $v=\valOf{o}$. Then $o$ is the last enqueue in $\alpha$, and it is the
unique unmatched enqueue from $\alpha$: any earlier unmatched enqueue would
contradict the choice of $o$, and any later enqueue lies in $\beta$.
Moreover, any successful $\peek$ or $\deq$ in $\beta$ observes or matches an
enqueue in $\alpha$, since $o$ is the first unmatched enqueue and remains the
front of the queue throughout the suffix. Thus $\Survives{v}(\alpha,\beta)$
holds. By \lemref{queue-rewrite-invariant}, $(\alpha,\beta)\in R_v$.
\end{proof}

\subsection{Proofs from \secref{queue-lang-reachability}}
\applabel{queue-lang-reachability}

\begin{lemma}[Enabled-edge commutation]\lemlabel{enabled-edge-commutation}
Let $\hist$ be a history and let $I_1,I_2,I_3 \in \fgnodes{\hist}$.
If $(I_1,o,I_2)\in\fgedges{\hist}$ and $I_2 \subseteq I_3$, then for every
path $I_1 \leadsto_w I_3$, there exist words $u,v$ such that
$w = u \cdot o \cdot v$ and $I_2 \leadsto_{u\cdot v} I_3$.
\end{lemma}

\begin{proof}
Since $(I_1,o,I_2)\in\fgedges{\hist}$, we have
$I_2 = I_1 \cup \set{o}$. As $I_2 \subseteq I_3$, the operation $o$
belongs to $I_3 \setminus I_1$. Hence every path from $I_1$ to $I_3$
must add $o$ exactly once. Let
\[
  I_1 = J_0 \leadsto_{\abs{o_1}} J_1
  \leadsto_{\abs{o_2}} \cdots
  \leadsto_{\abs{o_m}} J_m = I_3
\]
be a path labeled $w=\abs{o_1}\cdots \abs{o_m}$, and let $r$ be the
unique index such that $o_r=o$. Write
\[
  w = u \cdot o \cdot v
\]
where $u=\abs{o_1}\cdots \abs{o_{r-1}}$ and
$v=\abs{o_{r+1}}\cdots \abs{o_m}$.

We construct a path from $I_2$ to $I_3$ labeled $u\cdot v$ by taking the
edge adding $o$ first and then replaying the remaining operations in the
same order. For $0\leq i<r$, define
\[
  J'_i = J_i \cup \set{o}.
\]
We claim each $J'_i$ is an ideal. Indeed, since $I_2=I_1\cup\set{o}$ is
an ideal, all predecessors of $o$ already belong to $I_1$. Since
$I_1\subseteq J_i$, adding $o$ to $J_i$ preserves downward closure.
Thus each $J'_i$ is a frontier-graph node.

Moreover, for every $0\leq i<r-1$, the transition
$J'_i \to J'_{i+1}$ adds exactly the operation $o_{i+1}$, and hence is a
frontier-graph edge labeled $\abs{o_{i+1}}$. Thus we have a path
\[
  I_2 = J'_0
  \leadsto_{\abs{o_1}} J'_1
  \leadsto_{\abs{o_2}} \cdots
  \leadsto_{\abs{o_{r-1}}} J'_{r-1}.
\]
Since $J_r = J_{r-1}\cup\set{o}=J'_{r-1}$, we may then continue along
the original path from $J_r$ to $I_3$, using the edges labeled
$\abs{o_{r+1}},\ldots,\abs{o_m}$. Therefore
\[
  I_2 \leadsto_{u\cdot v} I_3.
\]
\end{proof}

\begin{lemma}[Sequence split]
\lemlabel{queue-sequence-split}
Let $\hist$ be a concurrent history and let
$\tau = o \cdot \beta \cdot o'$ be a linearization of $\hist$.
If $o$ and $o'$ are concurrent,
then there exists a decomposition $\beta = \beta_1 \cdot \beta_2$ such that
every operation in $\beta_1$ is concurrent with $o$ and every operation in
$\beta_2$ is concurrent with $o'$.
\end{lemma}

\begin{proof}
Assume for contradiction that no such decomposition exists.
Then there are two possibly non-distinct operations $x,y$ in $\beta$
such that $\tau(x) \leq \tau(y)$, and $o <_\hist x$ and $y <_\hist o'$.
We hence have $\resTimeOf{y} < \invTimeOf{o'} < \resTimeOf{o} < \invTimeOf{x}$.
Hence, $y <_\hist x$, which contradicts $\tau(x) \leq \tau(y)$.

The conclusion follows.
\end{proof}

\begin{lemma}[Soundness of queue ideal relations]
\lemlabel{queue-ideal-rel-sound}
Let $\hist$ be a concurrent $\queueDS{}$ history. If
$(I_1,I_2)\in \IdlRel{v_{\qbot}}$, then there exists a path
$I^{\init}_{\hist}\leadsto_{\tau} I_2$ such that there is a split
$\tau=\alpha\cdot\beta$ of $\tau$ for which the following hold:
\begin{enumerate}
  \item the enqueue operations in $\alpha$ are precisely the enqueue
  operations in $I_1$,
  \item every enqueue operation in $\beta$ is unmatched in $\tau$,
  \item every successful $\peek$ operation in $\beta$ observes an enqueue
  operation in $\alpha$, and every successful $\deq$ operation in $\beta$
  matches an enqueue operation in $\alpha$,
  \item among the enqueue operations in $\alpha$, all are matched in $\tau$
  except possibly the last enqueue operation in $\alpha$,
  \item if all enqueue operations in $\alpha$ are matched in $\tau$, then
  $v_{\qbot}=\qbot$; otherwise, the last enqueue operation in $\alpha$ is
  unmatched in $\tau$ and has value $v_{\qbot}$.
\end{enumerate}
\end{lemma}

\begin{proof}
We prove the claim by induction on the derivation of
$(I_1,I_2)\in\IdlRel{v_{\qbot}}$.

\smallskip
\noindent
\emph{Base case.}
For the rule
$(I^{\init}_{\hist},I^{\init}_{\hist})\in\IdlRel{\qbot}$, take
$\tau=\alpha=\beta=\epsilon$. All conditions hold vacuously.

\smallskip
\noindent
\emph{Inductive cases.}
We consider the last rule used in the derivation.

\begin{enumerate}
\item \emph{Failed operation.}
Suppose the last rule is
\[
(I_1,I_1)\in\IdlRel{\qbot}
\land
(I_1,\opr{m}{}{\vfail},I_2)\in E_\hist
\implies
(I_1,I_2)\in\IdlRel{\qbot},
\]
where $m\in\set{\peek,\deq}$.
By induction, there is a path
$I^{\init}_{\hist}\leadsto_{\tau} I_1$ and a split
$\tau=\alpha\cdot\beta$ satisfying the invariant for $\qbot$.
Since $\tau$ linearizes exactly $I_1$ and the enqueue operations in
$\alpha$ are precisely those in $I_1$, the suffix $\beta$ contains no
enqueue operations. Moreover, since the control is $\qbot$, all enqueue
operations in $\alpha$ are matched; hence the queue is empty after $\tau$.
Thus
\[
\tau'=\tau\cdot\opr{m}{}{\vfail}
\]
is a legal queue linearization of $I_2$. Taking
$\alpha'=\alpha$ and $\beta'=\beta\cdot\opr{m}{}{\vfail}$ preserves all
enqueue-related conditions, and the newly added failed operation introduces
no successful observation or match. Therefore the invariant holds for
$(I_1,I_2)\in\IdlRel{\qbot}$.

\item \emph{Peek front.}
Suppose the last rule is
\[
(I_1,I_2)\in\IdlRel{v}
\land
(I_2,\opr{\peek}{}{v},I_3)\in E_\hist
\implies
(I_1,I_3)\in\IdlRel{v}.
\]
By induction, there is a path
$I^{\init}_{\hist}\leadsto_{\tau} I_2$ and a split
$\tau=\alpha\cdot\beta$ satisfying the invariant for $v$.
The last enqueue in $\alpha$ is unmatched in $\tau$ and has value $v$;
all enqueue operations before it in $\alpha$ are matched, and every enqueue
in $\beta$ is unmatched. Hence the front of the queue after $\tau$ has value
$v$, so
\[
\tau'=\tau\cdot\opr{\peek}{}{v}
\]
is legal. Take $\alpha'=\alpha$ and
$\beta'=\beta\cdot\opr{\peek}{}{v}$. No enqueue is added or matched, and the
new successful $\peek$ observes the unmatched enqueue from $\alpha$. Hence
the invariant holds for $(I_1,I_3)\in\IdlRel{v}$.

\item \emph{Dequeue front.}
Suppose the last rule is
\[
(I_1,I_2)\in\IdlRel{v}
\land
(I_2,\opr{\deq}{}{v},I_3)\in E_\hist
\implies
(I_1,I_3)\in\IdlRel{\qbot}.
\]
By induction, there is a path
$I^{\init}_{\hist}\leadsto_{\tau} I_2$ and a split
$\tau=\alpha\cdot\beta$ satisfying the invariant for $v$.
The last enqueue in $\alpha$ is unmatched in $\tau$ and has value $v$;
therefore appending $\opr{\deq}{}{v}$ is legal and matches this enqueue.
Let
\[
\tau'=\tau\cdot\opr{\deq}{}{v},\qquad
\alpha'=\alpha,\qquad
\beta'=\beta\cdot\opr{\deq}{}{v}.
\]
The enqueues in $\beta'$ are exactly those in $\beta$, and remain
unmatched; the new successful dequeue matches an enqueue from $\alpha$.
All enqueue operations in $\alpha$ are now matched in $\tau'$, so the
control becomes $\qbot$. Hence the invariant holds.

\item \emph{Deferred enqueue.}
Suppose the last rule is
\[
(I_1,I_2)\in\IdlRel{v_{\qbot}}
\land
(I_2,\opr{\enq}{v}{\vok},I_3)\in E_\hist
\implies
(I_1,I_3)\in\IdlRel{v_{\qbot}}.
\]
By induction, there is a path
$I^{\init}_{\hist}\leadsto_{\tau} I_2$ and a split
$\tau=\alpha\cdot\beta$ satisfying the invariant for $v_{\qbot}$.
Appending an enqueue is always legal, so let
\[
\tau'=\tau\cdot\opr{\enq}{v}{\vok},\qquad
\alpha'=\alpha,\qquad
\beta'=\beta\cdot\opr{\enq}{v}{\vok}.
\]
The enqueue operations in $\alpha'$ are unchanged, and the new enqueue is
in the suffix and unmatched. No successful $\peek$ or $\deq$ is introduced,
and the matched/unmatched status of enqueues in $\alpha'$ is unchanged.
Thus the invariant holds with the same control value $v_{\qbot}$.

\item \emph{Register front.}
Suppose the last rule is
\[
(I_1,I_3)\in\IdlRel{\qbot}
\land
(I_1,\opr{\enq}{v}{\vok},I_2)\in E_\hist
\land
I_1\subseteq I_2\subseteq I_3
\implies
(I_2,I_3)\in\IdlRel{v}.
\]
Let $e$ denote the enqueue operation labeling the edge
$(I_1,\opr{\enq}{v}{\vok},I_2)$.
By induction, there is a path
$I^{\init}_{\hist}\leadsto_{\tau} I_3$ and a split
$\tau=\alpha\cdot\beta$ satisfying the invariant for $\qbot$.
Thus the enqueue operations in $\alpha$ are precisely those in $I_1$, all
of them are matched in $\tau$, every enqueue in $\beta$ is unmatched, and
every successful $\peek$ or $\deq$ in $\beta$ observes or matches an enqueue
from $\alpha$.

Since $I_2=I_1\cup\set{e}\subseteq I_3$, the operation $e$ occurs in
$\tau$. As $e\notin I_1$ and the enqueue operations in $\alpha$ are exactly
those in $I_1$, the operation $e$ occurs in $\beta$.

We first transform the witness, if necessary, so that $e$ is the first
enqueue operation in the suffix $\beta$. If this is already the case, there
is nothing to do. Otherwise, let $e'$ be the enqueue operation immediately
preceding $e$ among the enqueue operations of $\beta$, and write the segment
from $e'$ to $e$ as
\[
e'\cdot \delta \cdot e ,
\]
where $\delta$ contains no enqueue operations. The operations $e'$ and $e$
are concurrent: since $e$ is enabled at $I_1$, all predecessors of $e$ are
already in $I_1$, whereas $e'\notin I_1$; and $e$ cannot precede $e'$ since
$e'$ appears before $e$ in the linearization $\tau$.

By \lemref{queue-sequence-split}, write $\delta=\delta_1\cdot\delta_2$
so that every operation in $\delta_1$ is concurrent with $e'$ and every
operation in $\delta_2$ is concurrent with $e$. Hence the segment
\[
e'\cdot \delta_1\cdot\delta_2\cdot e
\]
can be replaced by another linearization in which $e$ occurs before $e'$,
while preserving the relative order of all other enqueue operations.
Queue legality is preserved: the operations in $\delta$ are not enqueues;
successful $\peek$ and $\deq$ operations in the suffix observe or match
enqueues from $\alpha$, not from suffix enqueues; and failed operations in
the suffix, if any, occur before the first suffix enqueue and hence do not
occur in this segment. Thus all return values and matchings are preserved.

This swap decreases the number of enqueue operations of $\beta$ preceding
$e$. Repeating finitely many times, we obtain a legal linearization
$\widehat{\tau}=\widehat{\alpha}\cdot e\cdot\widehat{\beta}$ of $I_3$
such that:
(i) the enqueue operations in $\widehat{\alpha}$ are precisely those in
$I_1$, (ii) $e$ is the first enqueue operation after $\widehat{\alpha}$,
and (iii) all invariant conditions for the original $\qbot$-witness are
preserved.

Now set
\[
\alpha'=\widehat{\alpha}\cdot e,\qquad
\beta'=\widehat{\beta}.
\]
Then the enqueue operations in $\alpha'$ are precisely those in
$I_2=I_1\cup\set{e}$. All old enqueue operations from $\widehat{\alpha}$
are matched, while $e$ is unmatched, has value $v$, and is the last enqueue
operation in $\alpha'$. All enqueue operations in $\beta'$ are unmatched,
and every successful $\peek$ or $\deq$ in $\beta'$ still observes or matches
an enqueue from $\alpha'\). Hence the invariant holds for
$(I_2,I_3)\in\IdlRel{v}$.

\item \emph{Skip over peek.}
Suppose the last rule is
\[
(I_1,I_3)\in\IdlRel{v_{\qbot}}
\land
(I_1,\opr{\peek}{}{v_{\vfail}},I_2)\in E_\hist
\land
I_1\subseteq I_2\subseteq I_3
\implies
(I_2,I_3)\in\IdlRel{v_{\qbot}}.
\]
By induction, there is a path
$I^{\init}_{\hist}\leadsto_{\tau} I_3$ and a split
$\tau=\alpha\cdot\beta$ satisfying the invariant for $v_{\qbot}$.
Since the edge from $I_1$ to $I_2$ is not an enqueue, $I_1$ and $I_2$
contain the same enqueue operations. Therefore the same witness
$\tau=\alpha\cdot\beta$ satisfies the enqueue-agreement condition for
$I_2$. All other conditions are unchanged, so the invariant holds for
$(I_2,I_3)\in\IdlRel{v_{\qbot}}$.

\item \emph{Skip over dequeue.}
The argument is identical. The rule is
\[
(I_1,I_3)\in\IdlRel{v_{\qbot}}
\land
(I_1,\opr{\deq}{}{v_{\vfail}},I_2)\in E_\hist
\land
I_1\subseteq I_2\subseteq I_3
\implies
(I_2,I_3)\in\IdlRel{v_{\qbot}}.
\]
By induction, take a witness $\tau=\alpha\cdot\beta$ for
$(I_1,I_3)\in\IdlRel{v_{\qbot}}$. Since the added operation is not an
enqueue, $I_1$ and $I_2$ have the same enqueue operations. Thus the same
split witnesses the invariant for $(I_2,I_3)\in\IdlRel{v_{\qbot}}$.
\end{enumerate}
\end{proof}

\begin{lemma}[Completeness of queue ideal relations]
\lemlabel{queue-ideal-rel-complete}
Let $\hist$ be a concurrent $\queueDS{}$ history.
Suppose
\[
I^{\init}_{\hist} \leadsto_{\alpha} I_1
\qquad\text{and}\qquad
I_1 \leadsto_{\beta} I_2
\]
are paths in $\fg{\hist}$, and $(\alpha,\beta)\in R_{v_{\qbot}}$.
Then $(I_1,I_2)\in \IdlRel{v_{\qbot}}$.
\end{lemma}

\begin{proof}
We prove the claim by induction on the derivation of
$(\alpha,\beta)\in R_{v_{\qbot}}$. The induction is over all choices of
ideals $I_1,I_2$ witnessing paths
$I^{\init}_{\hist}\leadsto_{\alpha} I_1$ and
$I_1\leadsto_{\beta} I_2$.

\smallskip
\noindent
\emph{Base case.}
For the rule $(\epsilon,\epsilon)\in R_{\qbot}$, the path assumptions imply
$I_1=I_2=I^{\init}_{\hist}$. Hence
$(I^{\init}_{\hist},I^{\init}_{\hist})\in \IdlRel{\qbot}$ by the base rule of
\defref{queue-ideal-rel}.

\smallskip
\noindent
\emph{Inductive cases.}
We consider the last rule used in the derivation.

\begin{enumerate}
\item \emph{Failed operation.}
Suppose the last rule is
\[
(\alpha,\epsilon)\in R_{\qbot}
\Rightarrow
(\alpha,\opr{m}{}{\vfail})\in R_{\qbot},
\qquad m\in\set{\peek,\deq}.
\]
Assume
$I^{\init}_{\hist}\leadsto_{\alpha} I_1$ and
$I_1\leadsto_{\opr{m}{}{\vfail}} I_2$.
By induction applied to $(\alpha,\epsilon)\in R_{\qbot}$, we have
$(I_1,I_1)\in \IdlRel{\qbot}$. Since
$(I_1,\opr{m}{}{\vfail},I_2)\in E_{\hist}$, the corresponding failed-operation
rule of \defref{queue-ideal-rel} yields
$(I_1,I_2)\in \IdlRel{\qbot}$.

\item \emph{Peek front.}
Suppose the last rule is
\[
(\alpha,\beta)\in R_v
\Rightarrow
(\alpha,\beta\cdot\opr{\peek}{}{v})\in R_v.
\]
Assume
$I^{\init}_{\hist}\leadsto_{\alpha} I_1$ and
$I_1\leadsto_{\beta\cdot\opr{\peek}{}{v}} I_3$.
Factor the second path as
\[
I_1\leadsto_{\beta} I_2
\leadsto_{\opr{\peek}{}{v}} I_3.
\]
By induction, $(I_1,I_2)\in \IdlRel{v}$. Applying the peek-front rule of
\defref{queue-ideal-rel} gives $(I_1,I_3)\in \IdlRel{v}$.

\item \emph{Dequeue front.}
Suppose the last rule is
\[
(\alpha,\beta)\in R_v
\Rightarrow
(\alpha,\beta\cdot\opr{\deq}{}{v})\in R_{\qbot}.
\]
Assume
$I^{\init}_{\hist}\leadsto_{\alpha} I_1$ and
$I_1\leadsto_{\beta\cdot\opr{\deq}{}{v}} I_3$.
Factor the second path as
\[
I_1\leadsto_{\beta} I_2
\leadsto_{\opr{\deq}{}{v}} I_3.
\]
By induction, $(I_1,I_2)\in \IdlRel{v}$. Applying the dequeue-front rule of
\defref{queue-ideal-rel} gives $(I_1,I_3)\in \IdlRel{\qbot}$.

\item \emph{Deferred enqueue.}
Suppose the last rule is
\[
(\alpha,\beta)\in R_{v_{\qbot}}
\Rightarrow
(\alpha,\beta\cdot\opr{\enq}{v}{\vok})\in R_{v_{\qbot}}.
\]
Assume
$I^{\init}_{\hist}\leadsto_{\alpha} I_1$ and
$I_1\leadsto_{\beta\cdot\opr{\enq}{v}{\vok}} I_3$.
Factor the second path as
\[
I_1\leadsto_{\beta} I_2
\leadsto_{\opr{\enq}{v}{\vok}} I_3.
\]
By induction, $(I_1,I_2)\in \IdlRel{v_{\qbot}}$. Applying the deferred-enqueue
rule of \defref{queue-ideal-rel} gives
$(I_1,I_3)\in \IdlRel{v_{\qbot}}$.

\item \emph{Register front.}
Suppose the last rule is
\[
(\alpha,\opr{\enq}{v}{\vok}\cdot\beta)\in R_{\qbot}
\Rightarrow
(\alpha\cdot\opr{\enq}{v}{\vok},\beta)\in R_v.
\]
Assume
\[
I^{\init}_{\hist}\leadsto_{\alpha\cdot\opr{\enq}{v}{\vok}} I_2
\qquad\text{and}\qquad
I_2\leadsto_{\beta} I_3.
\]
Factor the first path as
\[
I^{\init}_{\hist}\leadsto_{\alpha} I_1
\leadsto_{\opr{\enq}{v}{\vok}} I_2.
\]
Then
\[
I_1\leadsto_{\opr{\enq}{v}{\vok}\cdot\beta} I_3.
\]
By induction applied to
$(\alpha,\opr{\enq}{v}{\vok}\cdot\beta)\in R_{\qbot}$, we obtain
$(I_1,I_3)\in \IdlRel{\qbot}$. Since
$(I_1,\opr{\enq}{v}{\vok},I_2)\in E_{\hist}$ and
$I_1\subseteq I_2\subseteq I_3$, the register-front rule of
\defref{queue-ideal-rel} yields $(I_2,I_3)\in \IdlRel{v}$.

\item \emph{Skip over peek.}
Suppose the last rule is
\[
(\alpha,\opr{\peek}{}{v_{\vfail}}\cdot\beta)\in R_{v_{\qbot}}
\Rightarrow
(\alpha\cdot\opr{\peek}{}{v_{\vfail}},\beta)\in R_{v_{\qbot}}.
\]
Assume
\[
I^{\init}_{\hist}\leadsto_{\alpha\cdot\opr{\peek}{}{v_{\vfail}}} I_2
\qquad\text{and}\qquad
I_2\leadsto_{\beta} I_3.
\]
Factor the first path as
\[
I^{\init}_{\hist}\leadsto_{\alpha} I_1
\leadsto_{\opr{\peek}{}{v_{\vfail}}} I_2.
\]
Then
\[
I_1\leadsto_{\opr{\peek}{}{v_{\vfail}}\cdot\beta} I_3.
\]
By induction applied to
$(\alpha,\opr{\peek}{}{v_{\vfail}}\cdot\beta)\in R_{v_{\qbot}}$, we obtain
$(I_1,I_3)\in \IdlRel{v_{\qbot}}$. Applying the skip-over-peek rule of
\defref{queue-ideal-rel} gives $(I_2,I_3)\in \IdlRel{v_{\qbot}}$.

\item \emph{Skip over dequeue.}
Suppose the last rule is
\[
(\alpha,\opr{\deq}{}{v_{\vfail}}\cdot\beta)\in R_{v_{\qbot}}
\Rightarrow
(\alpha\cdot\opr{\deq}{}{v_{\vfail}},\beta)\in R_{v_{\qbot}}.
\]
Assume
\[
I^{\init}_{\hist}\leadsto_{\alpha\cdot\opr{\deq}{}{v_{\vfail}}} I_2
\qquad\text{and}\qquad
I_2\leadsto_{\beta} I_3.
\]
Factor the first path as
\[
I^{\init}_{\hist}\leadsto_{\alpha} I_1
\leadsto_{\opr{\deq}{}{v_{\vfail}}} I_2.
\]
Then
\[
I_1\leadsto_{\opr{\deq}{}{v_{\vfail}}\cdot\beta} I_3.
\]
By induction applied to
$(\alpha,\opr{\deq}{}{v_{\vfail}}\cdot\beta)\in R_{v_{\qbot}}$, we obtain
$(I_1,I_3)\in \IdlRel{v_{\qbot}}$. Applying the skip-over-dequeue rule of
\defref{queue-ideal-rel} gives $(I_2,I_3)\in \IdlRel{v_{\qbot}}$.
\end{enumerate}
\end{proof}

\QueueReachability*

\begin{proof}
$(\Rightarrow)$ Suppose
$(\fg{\hist}, I^{\init}_{\hist}, I^{\final}_{\hist})$ satisfies
$\QueueSpec$-reachability. Then there is a path
\[
I^{\init}_{\hist} \leadsto_{\tau} I^{\final}_{\hist}
\]
such that $\tau \in \QueueSpec{}$.
By \lemref{queue-rewrite-correctness}, there exist
$\alpha,\beta$ and $v_{\qbot}\in \vals_\hist\uplus\set{\qbot}$ such that
$\tau=\alpha\cdot\beta$ and $(\alpha,\beta)\in R_{v_{\qbot}}$.
Let $I'$ be the frontier-graph node reached after reading the prefix
$\alpha$ along the above path. Then
\[
I^{\init}_{\hist}\leadsto_{\alpha} I'
\qquad\text{and}\qquad
I'\leadsto_{\beta} I^{\final}_{\hist}.
\]
By \lemref{queue-ideal-rel-complete},
$(I',I^{\final}_{\hist})\in \IdlRel{v_{\qbot}}$.

\smallskip
\noindent
$(\Leftarrow)$ Suppose there exist
$I'\in\fgnodes{\hist}$ and
$v_{\qbot}\in \vals_\hist\uplus\set{\qbot}$ such that
$(I',I^{\final}_{\hist})\in \IdlRel{v_{\qbot}}$.
By \lemref{queue-ideal-rel-sound}, there is a path
\[
I^{\init}_{\hist}\leadsto_{\tau} I^{\final}_{\hist}
\]
such that $\tau\in\QueueSpec{}$. Hence
$(\fg{\hist},I^{\init}_{\hist},I^{\final}_{\hist})$ satisfies
$\QueueSpec$-reachability.
\end{proof}

\end{document}
\endinput